%% file: double.tex
\documentclass[journal,letterpaper,twocolumn,final,10pt]{IEEEtran}
\usepackage{color}
\usepackage[final]{graphicx}
\usepackage{amssymb}
\usepackage{amsmath}
\usepackage{amsthm}
\usepackage{dsfont}
\usepackage{cite}
\input{Defs2}

\title{Adaptive Radar Detection of a Subspace Signal Embedded in Subspace  
Structured plus Gaussian Interference Via Invariance}

\vspace{0.1cm}

\author{A. De Maio, \IEEEmembership{Fellow, IEEE,} and D. Orlando, \IEEEmembership{Senior Member, IEEE}
\thanks{A. De Maio is with the Dipartimento di Ingegneria Elettrica e delle Tecnologie dell'Informazione, Universit\`a degli Studi di Napoli ``Federico II'',
via Claudio 21, I-80125 Napoli, Italy. E-mail: {\tt ademaio@unina.it}.}
\thanks{D. Orlando is with ELETTRONICA S.p.A., Via Tiburtina Valeria Km 13,700, 00131 Roma, Italy. Phone: +39 06 415 4873. E-mail: {\tt danilor78@gmail.com}.}
}

\begin{document}

\maketitle

\begin{abstract}
This paper deals with adaptive radar detection of a subspace signal  
competing with two sources of interference. The former is Gaussian  
with unknown covariance matrix and accounts for the joint presence of  
clutter plus thermal noise. The latter is structured as a subspace  
signal and models coherent pulsed jammers impinging on the radar  
antenna. The problem is solved via the {\em Principle of Invariance} which  
is based on the identification of a suitable group of transformations  
leaving the considered hypothesis testing problem invariant. A {\em maximal  
invariant statistic}, which completely characterizes the class of  
invariant decision rules and significantly compresses the original  
data domain, as well as its statistical characterization are  
determined. Thus, the existence of the optimum invariant detector is  
addressed together with the design of practically implementable  
invariant decision rules. At the analysis stage, the performance of  
some receivers belonging to the new invariant class is established  
through the use of analytic expressions.
\end{abstract}


\section{Introduction}

\PARstart{A}daptive radar detection of targets embedded in Gaussian interference is an active research field, which has enjoyed the 
effort of world class scientists. As a result, in the last decades, a plethora of adaptive or nonadaptive
detection strategies have been proposed
to face with realistic applications. Generally speaking, the existing architectures can be classified into two macro-classes which
differ in the level of a priori knowledge about the target response. The former assumes perfectly known target signature
at the design stage (rank-1 receivers) \cite[and references therein]{Kelly-GLRT,Kelly-Nitzberg,yuri,orlando2010,Pulsone,Bandiera07}, 
whereas the latter takes into account 
possible uncertainties in the nominal steering vector \cite{BOR-Morgan,DeMaioMismatches,Gini-Farina-Greco}. 
As a matter of fact, in situations of practical interest, several error sources (such as array calibration uncertainties, 
beampointing errors, multipath, etc.) \cite{antennaBased} can give rise to mismatched signals which can degrade the detection performances of
the architectures belonging to the first class \cite{BOR-Morgan,Richmond-1,Richmond-2}.

The subspace detection represents an effective analytic tool incorporating the partial 
knowledge of the target response in the detector
design and, hence, mitigating the performance degradation due to steering vector errors 
\cite{Kelly-TR,Scharf-Friedlander,ABGR-Subspace,BBORS-Direction,raghavanSubspace,BSV-Subspace,gini1,Fabrizio-Farina,Li2}.
The main idea of the subspace approach is to constrain the 
target steering vector to belong to a suitable subspace of the observation space. 
By doing so, it is possible to capture the energy of the potentially distorted wavefront of a mainlobe target. As a result,
architectures exploiting this trick are capable of declaring the presence of targets whose signature
is significantly different from the nominal one.
It is worth pointing out that the subspace idea can be also exploited to model coherent interfering signals
impinging on the radar antenna, whose directions of arrival have been estimated within some uncertainty.
For instance, in \cite{ABGR-Subspace,BBORS-Direction} the authors 
devise adaptive decision schemes to reveal
extended targets when the interference comprises a random contribution (referring to clutter and thermal noise) and
a structured unwanted component.

In \cite{raghavanSubspace}, \cite{Bose-Steinhardt_02}, and \cite{Scharf-Friedlander} the subspace detection 
has been dealt with the
invariance theory in hypothesis testing problems \cite{Muirhead,lehmann,Scharf-book}.
By doing so, the authors focus on decision rules exhibiting some natural symmetries
implying important practical properties such as the Constant False Alarm Rate (CFAR) behavior. Besides,
the use of invariance leads to a data reduction because all invariant tests can be expressed in terms of
a statistic, called {\em maximal invariant}, which organizes the original data into equivalence classes. Also the
parameter space is usually compressed after reduction by invariance and the dependence on the original
set of parameters becomes embodied into a maximal invariant in the parameter space (induced maximal
invariant). 
Further examples of the 
application of such theory to radar detection problems can be also found in \cite{Bose-Steinhardt,
Kay-Gabriel,DeMaio2,DeMaioConte_Invariance,Raghavan,ScharfInv,BessonInv,Raghavan-2}.

It is now important to observe that some of the just mentioned papers do not assume the 
presence of structured interference \cite{raghavanSubspace,Bose-Steinhardt_02,Bose-Steinhardt,Kay-Gabriel,
DeMaio2,DeMaioConte_Invariance,Raghavan,ScharfInv,BessonInv,Raghavan-2}, while 
others exploit the {\em Principle of Invariance}
to solve the nonadaptive subspace detection problem when a coherent interference source illuminates the radar antenna \cite{Scharf-Friedlander}.
The possibility to address the adaptive subspace detection problem in the joint presence of random and 
subspace structured interference via invariance, 
to the authors best knowledge, has not yet been considered in open literature. 
In this respect, the scope of this paper is to fill this gap.
To this end, this work is focused on adaptive detection of a subspace signal  
competing with two interference sources. The former is a completely  
random component, modeled as a Gaussian vector with unknown covariance  
matrix, and represents the returns from clutter and thermal noise. The  
latter is a subspace structured signal (with unknown location
parameters) and accounts for the presence of (possible) multiple  
pulsed coherent jammers impinging on the radar antenna from some  
directions. The problem is analytically formulated as a binary  
hypothesis test and the Principle of Invariance is exploited at the  
design stage to concentrate the attention on radar detectors enjoying  
some desirable practical features. More specifically, a suitable group  
of transformation leaving the considered testing problem unaltered is  
established. Hence, the practical importance of the resulting group  
action is explained as a mean to impose the CFAR property with respect to the clutter plus noise covariance  
matrix and the jammer location parameters. Otherwise stated, all the  
receivers sharing invariance with respect to the devised group of  
transformations exhibit the mentioned CFAR features.
Due to the primary importance of characterizing the family of invariant  
detectors, a maximal invariant statistic \cite{Muirhead,lehmann}, organizing the data into  
orbits and reducing the size of the observation space, is determined  
together with its statistical characterization. All the decision rules  
in the invariance family can be cast as functions of a maximal  
invariant. The problem of synthesizing the Most Powerful Invariant  
(MPI) detector \cite{lehmann} is also addressed and a discussion on the existence of  
the Uniformly MPI (UMPI) test is provided. Interestingly, the UMPI  
detector does not exist but for the case where the dimension of the  
useful signal plus jammer subspace equals the size of the  
observation space. In this specific situation, the energy detector  
turns out to be UMPI and coincides with the GLRT and the Locally MPI Detector (LMPID). When the UMPI test  
does not exist, the Conditional MPI Detector (CMPID), obtained as the  
optimum detector conditioned on the ancillary part of the maximal  
invariant is designed \cite{lehmann}, showing that it is statistically equivalent to  
the GLRT.
At the analysis stage, we investigate the performances of some practically  
implementable invariant detectors providing analytical  
equations for the computation of the false alarm and the detection  
probabilities. The performance loss with respect to the clairvoyant  
MPI test is finally assessed showing that it is an acceptable value in  
some instances of practical interest.

The paper is organized as follows. The next section contains the definitions used in the mathematical derivations, while Section III is devoted
to problem formulation and derivation of the maximal invariant statistic. In Section IV, we provide the statistical distribution of the
maximal invariant, whereas in Section V, invariant decision schemes are devised and assessed. Some concluding 
remarks and hints for future work are given in Section VI. 
Finally, Appendices \ref{appendix:InvariantTransformations}-\ref{appendix:ProofOfLMPIDdet} contain 
analytical derivations of the presented results.

\textcolor{black}{
\subsection{Notation}
In the sequel, vectors and matrices are denoted by boldface lower-case and upper-case letters, respectively. 
Symbols $\det(\cdot)$ and $\mbox{Tr}(\cdot)$
denote the determinant and the trace of a square matrix, respectively. If $A$ and $B$ are scalars, then $A\times B$ is the usual product of scalars;
on the other hand, if $A$ and $B$ are generic sets, $A \times B$ denotes the Cartesian product of sets.
The imaginary unit is $j$, i.e., $\sqrt{-1}=j$. We denote by $\bI_N$ the
$N\times N$-dimensional identity matrix, while $\bzero$ is the null vector or matrix of proper dimensions. The Euclidean norm of a vector is denoted
by $\|\cdot\|$. As to the numerical sets, $\N$ is the set of natural numbers, $\R$ is the set of real numbers,
$\R^{N\times M}$ is the set of $(N\times M)$-dimensional real matrices (or vectors if $M=1$),
$\C$ is the set of complex numbers, and $\C^{N\times M}$ is the set of $(N\times M)$-dimensional complex matrices (or vectors if $M=1$).
The subspace spanned by the columns of a matrix $\bA\in\C^{m\times n}$ is denoted by $\langle\bA\rangle$.
The real and imaginary parts of a complex vector or scalar are denoted by $\Re(\cdot)$ and $\Im(\cdot)$, respectively.
Symbols $(\cdot)^*$, $(\cdot)^T$, and $(\cdot)^\dag$ stand for complex conjugate, transpose, and conjugate transpose, respectively.
The acronym iid means independent and identically distributed.
If $\bx\in\C^{N\times 1}$ is
distributed according to the complex normal distribution with mean $\bm\in\C^{N\times 1}$
and covariance matrix $\bM\in\C^{N\times N}$, we write $\bx\sim\cC\cN_N(\bm, \bM)$. 
The complex noncentral F-distribution with $n,m$ complex degrees of freedom and noncentrality parameter $\delta$ is denoted by
$\cC\cF_{n,m}(\delta)$. Symbol $\cC\cF_{n,m}(0)$ is used to indicate the complex central F-distribution.
The complex noncentral beta distribution with $n,m$ complex degrees of freedom and noncentrality parameter $\delta$ is denoted by
$\cC\beta_{n,m}(\delta)$, whereas we refer to the complex central beta distribution using symbol $\cC\beta_{n,m}(0)$.
Symbol $\bS\sim\cC\cW_N(K,\bA)$ means that $\bS\in\C^{N\times N}$ obeys the complex central Wishart distribution with
parameters $K\in\N$ and $\bA\in\C^{N\times N}$ positive definite matrix. Finally,
if $x\in\C$ obeys the complex noncentral chi-square distribution with $K$ complex degrees of freedom and noncentrality parameter $\delta$, we
write $x\sim\cC\chi^2_{K}(\delta)$; on the other hand, if $x\in\C$ is ruled by the complex central chi-square distribution 
with $K$ complex degrees of freedom, we write $x\sim\cC\chi^2_{K}(0)$.
}

\section{Problem Formulation and Maximal Invariants}

In this section, we describe the detection problem at hand and introduce the Principle of Invariance that allows
a compression of data dimensionality. Assume that a sensing system collects data from $N\geq 2$ channels (spatial and/or temporal).
The returns from the cell under test, after pre-processing, are properly sampled
and organized to form a $N$-dimensional vector, $\bor$ say.
We want to test whether or not $\bor$ contains useful target echoes assuming the presence of an interfering signal. 

The target signature is modeled as a vector in a known subspace, $\langle\bH\rangle$ say, where $\bH\in\C^{N\times r}$, $r\geq 1$, is a full-column-rank
matrix. In other words, the useful echoes can be expressed as a linear combination of the columns of $\bH$, i.e., $\bH\bp$ with $\bp\in\C^{r\times 1}$.
On the other hand, the interference component consists of two contributions. The former is representative of the clutter echoes and
thermal noise, while the latter accounts for possible coherent sources impinging on the receive antenna from directions different to that
where the radar system is steered. More precisely, the structured interferer signal 
is assumed to belong to a known subspace, $\langle\bJ\rangle$ say, with $\bJ\in\C^{N\times t}$, $t\geq 1$, 
a full-column-rank matrix, and, hence, it is a linear combination, $\bq\in\C^{t\times 1}$ say, of the columns of $\bJ$ with $m=t+r\leq N$.
The above model (with reference to Doppler processing) 
comes in handy to deal with scenarios where the presence of one or multiple coherent pulsed jammers from standoff platforms attempt
to protect a target located
in the mainlobe of the radar antenna.
We further assume that the matrix $[\bJ \ \bH]\in\C^{N\times m}$ is full-column-rank, namely that the columns of $\bJ$ are linearly independent of
those of\footnote{It is worth pointing out that the case $t=0$ (i.e., absence of structured interference) has already been considered in open
literature \cite{raghavanSubspace}.} $\bH$. 
Finally, as customary, we suppose that a set of $K$ secondary data,
$\bor_k \in \C^{N\times 1}$, $k=1, \ldots, K$, $K \geq N$, namely data free of signal components and structured interference, is available.

Summarizing, the decision problem at hand can be formulated in terms of the binary hypothesis test
\be
\left\{
\begin{array}{ll}
H_0:\left\{\begin{array}{ll}
\bor = \bJ\bq + \bn_0, \phantom{\alpha \bs+\bs} &  \\
\bor_k = \bn_{0k}, & k=1, \ldots, K,
\end{array}\right.
\\
\\
H_1:\left\{\begin{array}{ll}
\bor = \bH \bp + \bJ \bq  +  \bn_0, &  \\
\bor_k = \bn_{0k}, & k=1, \ldots, K,
\end{array}\right.
\end{array}
\right.
\label{eq:hypothesistest00}
\ee
where
\begin{itemize}
\item $\bH\bp\in\C^{N\times 1}$ and $\bJ\bq\in\C^{N\times 1}$ are the target and interference signatures, respectively, with $\bp$ and $\bq$ deterministic and 
unknown vectors\footnote{Observe that $\bp$ also jointly accounts for target reflectivity, radar equation, and channel effects.};
\item $\bn_0$ and $\bn_{0k} \in \mathbb{C}^{N\times 1}$,  $k=1, \ldots, K$, are iid complex normal random vectors with zero mean
and unknown positive definite covariance matrix $\bM_0\in\C^{N\times N}$.
\end{itemize}
Let us now recast (\ref{eq:hypothesistest00}) in canonical form \cite{Kelly-TR}. To this end, 
we exploit the QR-decomposition of the partitioned matrix $[\bJ \ \bH]$ given by
\be
[\bJ \ \bH] = \bQ \bR,
\ee
where $\bQ\in\C^{N\times m}$ is a slice of unitary matrix (or a unitary matrix if $m=N$), 
i.e., $\bQ^\dag\bQ=\bI_{m}$ (or $\bQ^\dag\bQ=\bQ\bQ^\dag=\bI_{N}$ if $m=N$), and $\bR\in\C^{m\times m}$ is a nonsingular upper triangular matrix.
Moreover, observe that
\be
\bQ = [\bQ_J \ \bQ_0], \quad \bR = \left[\begin{array}{cc}
                                    \bR_J & \bR_0
                                    \\
                                    \bzero & \bR_1
                                   \end{array}\right],
                                   \label{eq:QRdecomposition}
\ee
where $\bQ_J\in\C^{N\times t}$ and $\bR_J\in\C^{t\times t}$ come from the QR-decomposition of $\bJ$, namely, 
$\bJ=\bQ_J\bR_J$ with $\bQ_J$ such that $\bQ_J^\dag\bQ_J=\bI_t$ and $\bR_J$ a nonsingular upper triangular matrix,
$\bR_0\in\C^{t \times r}$, $\bR_1\in\C^{r\times r}$ is another nonsingular upper triangular matrix.
Equalities (\ref{eq:QRdecomposition}) are almost evident consequences of the Gram-Schmidt procedure.
Now, let us define a unitary matrix $\bU\in\C^{N\times N}$ which rotates the columns of $\bQ$ onto the first $m$ elementary vectors
of the standard basis of $\C^{N\times 1}$, i.e.,
\be
\bU\bQ=\left[
\begin{array}{cc}
 \bI_t & \bzero
 \\
 \bzero & \bI_r
 \\
 \bzero & \bzero
\end{array}
\right]=[\bE_t \ \bE_r],
\ee
where\footnote{It is worth noticing that when $m=N$, $[\bE_t \ \bE_r]=\bI_N$.} $\bE_t=[\bI_t \ \bzero \ \bzero]^T$ and
$\bE_t=[\bzero \ \bI_r \ \bzero]^T$.
Thus, under $H_1$, primary and secondary data can be transformed as follows
\begin{align}
\bz &= \bU\bor 
= \bU \left\{
\left[ \bJ \ \bH \right] 
\left[
\begin{array}{c}
 \bq
 \\
 \bp
\end{array}
\right]+\bn_0\right\}\nonumber
\\
&=\bU \left\{
\bQ\left[ \bR_q \ \bR_p \right] 
\left[
\begin{array}{c}
 \bq
 \\
 \bp
\end{array}
\right]+\bn_0\right\}\nonumber
\\
&=\left[ \bE_t \ \bE_r \right](\bR_q\bq + \bR_p \bp)+\bn \nonumber
\\
&= \left[ \bE_t \ \bE_r \right]\left[
\begin{array}{c}
 \btheta_{11}
 \\
 \btheta_2
\end{array}
\right]+\bn = \bE_t\btheta_{11} + \bE_r\btheta_2 + \bn,
\\
\bz_k &= \bU\bor_k = \bn_k,\quad k=1,\ldots,K,
\end{align}
where $\bR_q=[\bR_j^T \ \bzero]^T$, $\bR_p=[\bR_0^T \ \bR_1^T]^T$, $\bn=\bU\bn_0$, $\bn_k=\bU\bn_{0k}$, $k=1,\ldots,K$,
$\btheta_{11}=\bR_J\bq+\bR_0\bp\in\C^{t\times 1}$, and $\btheta_2=\bR_1\bp\in\C^{r\times 1}$.
On the other hand, when $H_0$ is in force, the transformed data can be written as
\begin{align}
\bz &= \bU\bor=\bU\bJ\bq + \bn = \bE_t\bR_J\bq + \bn = \bE_t\btheta_{10}+\bn,
\\
\bz_k &= \bU\bor_k = \bn_k,
\end{align}
where $\btheta_{10}=\bR_J\bq\in\C^{t\times 1}$. 
Gathering the above results, we can rewrite problem (\ref{eq:hypothesistest00}) in canonical form
\be
\left\{
\begin{array}{ll}
H_0:\left\{\begin{array}{ll}
\bz = \bE_t\btheta_{10} + \bn, \phantom{\alpha \bs+\bs} &  \\
\bz_k = \bn_{k}, &  k=1, \ldots, K,
\end{array}\right.
\\
\\
H_1:\left\{\begin{array}{ll}
\bz = \bE_t \btheta_{11} + \bE_r \btheta_2  +  \bn, &  \\
\bz_k = \bn_{k}, &  k=1, \ldots, K,
\end{array}\right.
\end{array}
\right.
\label{eq:hypothesistest01}
\ee
where $\bn$, $\bn_k$, $k=1,\ldots,K$, are iid complex normal random vectors with zero mean and covariance matrix 
$\bM=\bU\bM_0\bU^\dag$.
Two important remarks are now in order. First, notice that $\langle\bH\rangle$ receives part of the interference energy and
that $\langle\bJ\rangle$ includes useful signal components since the columns $\bH$ and those of $\bJ$ are 
not orthogonal but only linearly independent.
This reciprocity becomes more evident after applying the transformations which lead to the canonical form.
Second, in problem (\ref{eq:hypothesistest00}) the relevant parameter to decide for the presence of a target is $\bp$. Otherwise stated,
if the $H_1$ hypothesis holds true, then $\|\bp\| > 0$, while $\|\bp\|=0$ under the disturbance-only hypothesis ($H_0$).
As a consequence, since $\bR_1$ is nonsingular, problem (\ref{eq:hypothesistest01}) is equivalent to
\be
\left\{
\begin{array}{ll}
H_0:\|\btheta_2\| = 0,
\\
H_1:\|\btheta_2\| > 0,
\end{array}
\right.
\ee
which partitions the parameter space, $\bTheta$ say, as
\be
\bTheta = \underbrace{\{ \bzero \}}_{\bTheta_0} \cup \underbrace{\{ \btheta_2\in\C^{r\times 1}: \ \| \btheta_2 \|>0 \}}_{\bTheta_1}.
\ee
In the following, we can look for decision rules sharing some invariance with respect 
to those parameters (nuisance parameters, in this case $\bM$, $\btheta_{11}$ 
under $H_1$, and $\btheta_{10}$ under $H_0$) which are irrelevant to the specific decision problem.
To this end, we exploit the Principle of Invariance \cite{lehmann}, whose main idea consists
in finding transformations that properly cluster data without altering
\begin{itemize}
\item the formal structure of the hypothesis testing problem given by $H_0:\|\btheta_2\| = 0$, $H_1:\|\btheta_2\| > 0$;
\item the Gaussian model assumption under the two hypotheses;
\item the subspace containing the useful signal components.
\end{itemize}
Before introducing the transformation group that fulfills the above requirements, let us define
$\bZ_s = [\bz_{1} \ \ldots \ \bz_{K}]$,
\be
\bz = \left[
\begin{array}{c}
\bz_{1}
\\
\bz_{2}
\\
\bz_{3}
\end{array}
\right], \ 
\bS = \bZ_s\bZ_s^T=\left[
\begin{array}{ccc}
\bS_{11} & \bS_{12} & \bS_{13}
\\
\bS_{21} & \bS_{22} & \bS_{23}
\\
\bS_{31} & \bS_{32} & \bS_{33}
\end{array}
\right],
\label{eq:partition}
\ee
where $\bz_{1}\in\C^{t \times 1}$, $\bz_{2}\in\C^{r\times 1}$, $\bz_{3}\in\C^{(N-m)\times 1}$, and
$\bS_{ij}$, $(i,j)\in\{1,2,3\}\times\{1,2,3\}$, is a submatrix whose dimensions can be obtained replacing 1, 2, and 3 with $t$, $r$, and $N-m$, respectively\footnote{Hereafter, in the case $m=N$, the ``$3$-components'' are no longer present in the partitioning.}.
Moreover, Fisher-Neyman factorization theorem ensures that deciding from the raw data $(\bz,\bZ_s)$ is tantamount to
deciding from a sufficient statistic $(\bz, \bS)$ \cite{Muirhead}.

Now, denote by $\cG\cL(N)$ the linear group of the $N\times N$ nonsingular matrices and introduce the following sets
\begin{multline}
\cG=\Bigg\{
\bG = \left[
\begin{array}{ccc}
\bG_{11} & \bG_{12} & \bG_{13}
\\
\bzero & \bG_{22} & \bG_{23}
\\
\bzero & \bzero & \bG_{33}
\end{array}
\right]\in\cG\cL(N): 
\\
\bG_{11}\in\cG\cL(t), \ \bG_{22}\in\cG\cL(r), \ \bG_{33}\in\cG\cL(N-m)
\Bigg\}
\end{multline}
and
\be
\cF=\left\{
\bff = \left[
\begin{array}{c}
\bff_{11} 
\\
\bzero
\end{array}
\right]\in\C^{N\times 1}: \bff_{11}\in\C^{t\times 1}
\right\}.
\ee
It follows that the set $\cG\times\cF$ together with the operation $\circ$ defined by
\be
(\bG_1,\bff_1)\circ(\bG_2,\bff_2) = (\bG_2\bG_1,\bG_2\bff_1+\bff_2),
\ee
form a group, $\cL$ say\footnote{Observe that 
\begin{itemize}
\item $\cL$ is closed with respect to the operation $\circ$;
\item $\forall (\bG_1,\bff_1)$, $(\bG_2,\bff_2)$, and $(\bG_3,\bff_3)\in\cL$: 
$[(\bG_1,\bff_1)\circ(\bG_2,\bff_2)]\circ(\bG_3,\bff_3)=(\bG_1,\bff_1)\circ[(\bG_2,\bff_2)\circ(\bG_3,\bff_3)]$ (Associative Property);
\item there exists and it is a unique $(\bG_I,\bff_I)\in\cL$ such that $\forall (\bG,\bff)\in\cL$: $(\bG_I,\bff_I)\circ(\bG,\bff)=(\bG,\bff)\circ(\bG_I,\bff_I)
=(\bG,\bff)$ (Identity Element);
\item $\forall (\bG,\bff)\in\cL$ there exists $(\bG_{-1},\bff_{-1})\in\cL$ such that $(\bG_{-1},\bff_{-1})\circ(\bG,\bff)
=(\bG,\bff)\circ(\bG_{-1},\bff_{-1})=(\bG_I,\bff_I)$ (Inverse Element).
\end{itemize}}, 
that leaves the hypothesis testing problem (\ref{eq:hypothesistest01}) invariant under the action 
$l$ defined by
\be
l(\bz,\bS)=(\bG\bz+\bff,\bG\bS\bG^\dag), \quad \forall (\bG,\bff)\in\cL.
\ee
The proof of the above statement is given in Appendix \ref{appendix:InvariantTransformations}.
Moreover, it is important to point out that $\cL$ preserves the family of distributions 
and, at the same time, includes those transformations which are relevant from the practical point of view as they 
allow to claim the CFAR property (with respect to $\bM$ and $\bq$) as a consequence of the invariance.

Summarizing, we have identified a group of transformations $\cL$ which leaves unaltered
the decision problem under consideration.
As a consequence,
it is reasonable to find decision rules that are invariant under $\cL$.
Toward this goal, the Principle of Invariance is invoked because it
allows to construct statistics that organize data into distinguishable equivalence classes. Such functions
of data are called maximal invariant statistics and, given the group of transformations, every invariant test may be written
as a function of a maximal invariant \cite{Scharf-book}.

The following proposition provides the expression of a maximal invariant for the problem at hand.
To this end, it is necessary to distinguish between $m<N$ and $m=N$; the latter equality implies that the ``3-components'', namely
$\bz_3$, $\bS_{13}$, $\bS_{23}$, $\bS_{33}$, $\bS_{31}$, $\bS_{32}$, $\bG_{13}$, $\bG_{23}$, and $\bG_{33}$,
are no longer present in the partitioned data matrices and vectors. 

\begin{prop}
\label{proposition:derivationMaxInv}
A maximal invariant statistic with respect to $\cL$ for problem (\ref{eq:hypothesistest01}) is given by
\be
\bt_1(\bz,\bS)=
\left\{
\begin{array}{ll}
\left[
\begin{array}{c}
\bz_{2.3}^\dag\bS_{2.3}^{-1}\bz_{2.3}
\\
\bz_3^\dag\bS_{33}^{-1}\bz_3
\end{array}
\right] & m<N,
\\
\bz_{2}^\dag\bS_{22}^{-1}\bz_{2} & m=N,
\end{array}
\right.
\label{eq:maxInv_00}
\ee
where $\bz_{2.3}=\bz_2-\bS_{23}\bS_{33}^{-1}\bz_3$ and $\bS_{2.3}=\bS_{22}-\bS_{23}\bS_{33}^{-1}\bS_{32}$.
\end{prop}
\begin{proof}
See Appendix \ref{appendix:maximalInvariantsProof}.
\end{proof}
Some remarks are now in order.
\begin{itemize}
\item Observe that when $m<N$ the maximal invariant statistic is a 2-dimensional vector where the second component represents 
an ancillary part, whose distribution does not depend on which hypothesis is in force.
\item It is important to compare (\ref{eq:maxInv_00}) with the maximal invariant derived in \cite{raghavanSubspace} assuming the absence
of subspace structured interference. Interestingly, they share a similar structure/dimensionality but for the presence of
a projection operation (onto the orthogonal complement of the rotated $\bJ$) which amounts to cleaning the received data 
vectors from the structured interference.
\item Finally, note that (see {\em Theorem 6.2.1} of \cite{lehmann}) every invariant test may be written as a function 
of (\ref{eq:maxInv_00}) (see Figure \ref{fig:trasformationInvariantTest} for a schematic representation).
\end{itemize}

\section{Distribution of the Maximal Invariants}
In this section, we provide the statistical characterization of the maximal invariant for the case $m<N$ and, then, we give a corollary
referring to $m=N$. To this end,
firstly we show that the maximal invariant can be written as a function of whitened random vectors and matrices and then
we find a suitable stochastic representation by means of one-to-one transformations. 
Finally, the probability density functions (pdfs) in the invariant domain
will be used in Section IV to
construct the Likelihood Ratio Test (LRT) of $H_0$ versus $H_1$, also referred to as the MPI Detector (MPID).

Let us invoke the invariance principle and consider the following transformation $(\bW,\bff_0)\in\cL$
\begin{align}
&\bw=\bV\bz+\bff_0=\left[
\begin{array}{c}
\bw_1
\\
\bw_2
\\
\bw_3
\end{array}
\right] 
\\
&\bS_0=\bV\bS\bV^\dag=\left[
\begin{array}{ccc}
\bS_{011} & \bS_{012} & \bS_{013}
\\
\bS_{021} & \bS_{022} & \bS_{023}
\\
\bS_{031} & \bS_{032} & \bS_{033}
\end{array}
\right],
\end{align}
where
\be
\bV=\left[
\begin{array}{ccc}
\bV_{11} & \bV_{12} & \bV_{13}
\\
\bzero & \bV_{22} & -\bM_{22}\bM_{23}\bM_{33}^{-1}
\\
\bzero & \bzero & \bM_{33}^{-1/2}
\end{array}
\right].
\label{eq:whiteningW}
\ee
In (\ref{eq:whiteningW}), $\bV_{22}=(\bM_{22}-\bM_{23}\bM_{33}^{-1}\bM_{23}^\dag)^{-1/2}$, 
$\bV_{ij}$, $i=1$, $j=1,2,3$, are matrices of proper dimensions, while
$\bM_{22}\in\C^{r\times r}$, $\bM_{23}\in\C^{r\times (N-m)}$, and $\bM_{33}\in\C^{(N-m)\times (N-m)}$ 
are obtained partitioning $\bM$ in the same way as $\bS$.

Now, observe that under $H_i$, $i=0,1$,
\begin{align}
\left[
\begin{array}{c}
\bw_2
\\
\bw_3
\end{array}
\right]
&=\bV_2
\left[
\begin{array}{c}
\bz_2
\\
\bz_3
\end{array}
\right]\sim\cC\cN(i\bV_2\bE_r\btheta_2,\bI_{N-t})
\label{eq:distribuzione_w}
\\
\left[
\begin{array}{cc}
\bS_{022} & \bS_{023}
\\
\bS_{032} & \bS_{033}
\end{array}
\right]
&=\bV_2
\left[
\begin{array}{cc}
\bS_{22} & \bS_{23}
\\
\bS_{32} & \bS_{33}
\end{array}
\right]\bV_2^\dag\nonumber
\\
&\sim \cC\cW_{N-t}(K,\bI_{N-t}),
\label{eq:distribuzione_S0}
\end{align}
where
\be
\bV_2=\left[
\begin{array}{cc}
\bV_{22} & -\bM_{22}\bM_{23}\bM_{33}^{-1}
\\
\bzero & \bM_{33}^{-1/2}
\end{array}
\right].
\ee
Thus, we can recast $\bt_1(\bz,\bS)$ in terms of whitened quantities, more precisely, as
\begin{align}
&m_1=(\bz_2-\bS_{23}\bS_{33}^{-1}\bz_3)^\dag (\bS_{22}-\bS_{23}\bS_{33}^{-1}\bS_{32})^{-1} \nonumber
\\
&\times (\bz_2-\bS_{23}\bS_{33}^{-1}\bz_3) \nonumber
\\
&=(\bw_2-\bS_{023}\bS_{033}^{-1}\bw_3)^\dag (\bS_{022}-\bS_{023}\bS_{033}^{-1}\bS_{032})^{-1} \nonumber
\\
&\times (\bw_2-\bS_{023}\bS_{033}^{-1}\bw_3),
\\
&m_2=\bz_3^\dag\bS_{33}^{-1}\bz_3 = \bw_3^\dag\bS_{033}^{-1}\bw_3.
\end{align}
Finally, apply a one-to-one transformation to (\ref{eq:maxInv_00}) to obtain the following stochastic representation for the maximal invariant
\be
\bt_2(\bz,\bS)=\bt_2(\bw,\bS_0)=\left[
\begin{array}{c}
\ds\frac{1}{\ds1+\frac{m_1}{1+m_2}}
\\
\vspace{-4mm}
\\
\ds\frac{1}{\ds1+m_2}
\end{array}
\right]
=\left[
\begin{array}{c}
p_1
\\
p_2
\end{array}
\right].
\ee
Summarizing, we have all the elements to provide the probability density functions (pdfs) of the maximal invariant under both hypotheses.
\begin{prop}
\label{prop:distributionMAXINV}
The joint pdf of $p_1$ and $p_2$ under the $H_i$ hypothesis, $i=0,1$, and assuming $m<N$ has the following expression
\begin{multline}
f_{p_1,p_2}(x,y;H_i)=f_{\beta}(x;K-(N-t)+1,r,0)
\\
\times f_{\beta}(y;K-(N-t)+r+1,N-t-r,0)e^{-i \ \mbox{\scriptsize \em SINR}  \ x y}
\\
\times \sum_{k=0}^{K-(N-t)+1}\binom{K-(N-t)+1}{k}
\\
\times \frac{(r-1)!}{(r+k-1)!}[\mbox{\em SINR} \  y (1-x)]^k,
\end{multline}
where $(x,y)\in \{(0 \ 1] \times (0 \ 1] \}$,
\be
\mbox{\em SINR}=\btheta_2^\dag ( \bM_{22}-\bM_{23}\bM_{33}^{-1}\bM_{32} )^{-1} \btheta_2
\ee
is the Signal-to-Noise-plus-Interference Ratio (SINR), and
$f_{\beta}(x;n,m,\delta)$ is the pdf of a \cite{Kelly-TR,BOR-Morgan}
\begin{itemize}
\item a random variable following a complex noncentral beta distribution with $n,m$ complex degrees of freedom and noncentrality parameter $\delta$, 
if $\delta>0$;
\item a random variable following a complex central beta distribution with $n,m$ complex degrees of freedom, if $\delta=0$.
\end{itemize}
\end{prop}
\begin{proof}
See Appendix \ref{appendix:ProofStatisticalCharacterization}.
\end{proof}

\begin{corol}
In the case $m=N$, the pdf of $p_3=1/(1+\bz_2^\dag\bS_{22}^{-1}\bz_2)$ under the $H_i$ hypothesis, $i=0,1$, is given 
by\footnote{Observe that $p_3$ obeys the complex beta distribution with $K-r+1,r$ complex degrees of freedom.}
\begin{multline}
f_{p_3}(x;H_i)=f_{\beta}(x;K-r+1,r,0)e^{-i \ \mbox{\scriptsize \em SINR}  \ x}
\\
\times \sum_{k=0}^{K-r+1}\binom{K-r+1}{k}\frac{(r-1)!}{(r+k-1)!}[\mbox{\em SINR} \  (1-x)]^k,
\end{multline}
where $x\in (0 \ 1]$ and $\mbox{\em SINR}=\btheta_2^\dag \bM_{22}^{-1} \btheta_2$.
\end{corol}
\begin{proof}
It is an application of the results contained in the proof of {\em Proposition \ref{prop:distributionMAXINV}}.
\end{proof}

As final remark, the induced maximal invariant can be easily obtained observing that the pdf of $\bt_2(\bz,\bS)$ depends
on the unknown parameters only through the SINR.

\section{Design and Analysis of Invariant Tests}
The goal of this section is twofold. First of all, we derive invariant decision schemes based on the LRT and its variants. 
In particular, we focus on the MPID, the LMPID (which is derived under the low SINR regime), and the Conditional MPID (CMPID), which are investigated
for the case where the Uniformly MPID (UMPID) does not exist.
Second, we consider some sub-optimum receivers providing the expression of their decision statistics as functions of the maximal invariant.

\subsection{LRT-based Decision Schemes in the Invariant Domain}
In the sequel, we consider the cases $m<N$ and $m=N$.

\subsubsection{Case $m<N$}
As starting point of our analysis, we devise the MPID
that, according to the Neyman-Pearson criterion, is given by the LRT after data compression. 
Specifically, the MPID has the following 
expression
\begin{align}
\tMPID &= 
\left.\frac{f_{p_1,p_2}(x,y;H_1)}{f_{p_1,p_2}(x,y;H_0)}\right|_{x=p_1,y=p_2} \nonumber
\\
&=e^{-\mbox{\scriptsize SINR}  \ p_1 p_2}
\sum_{k=0}^{K-(N-t)+1}\binom{K-(N-t)+1}{k}\nonumber
\\
&\times \frac{(r-1)!}{(r+k-1)!}[\mbox{SINR} \  p_2 (1-p_1)]^k\test\eta
\label{eq:MPID}
\end{align}
where $\eta$ is the threshold\footnote{Hereafter, we use symbol $\eta$ to denote the generic detection threshold.} 
set according to the desired Probability of False Alarm ($P_{fa}$). It is important to observe that
receiver (\ref{eq:MPID}) is clairvoyant and, as a consequence, it is not of great practical interest in radar applications since
it requires the knowledge of the induced maximal invariant. 
Nevertheless, it still represents a noteworthy benchmark, i.e., an upper bound
to the detection performance achievable by any invariant receiver. Before proceeding with the exposition, it is worth pointing out 
that decision scheme (\ref{eq:MPID}) is formally analogous to that found in \cite{raghavanSubspace}, where the optimum detector is devised
assuming a subspace structured signal embedded in random interference only (no subspace interference). 
This is perfectly consistent with the observation that when an additive structured interference
is present, the MPID (\ref{eq:MPID}) is obtained after projecting pre-processed data in canonical form onto 
the orthogonal complement of the interference subspace whose contribution
is thus removed.

Since in many applications the performances become critical for low SINR values, the LMPID in the limit of zero 
SINR could be of interest. Following the lead of \cite{Poor} and \cite{ViaLMPI}, the LMPID is given by
\be
\tLMPID=\frac{\ds \left.\frac{\delta f_{p_1,p_2}(p_1,p_2;H_1)}
{\delta\mbox{SINR}}\right|_{\mbox{\scriptsize SINR}=0}}{f_{p_1,p_2}(p_1,p_2;H_0)}\test\eta.
\label{eq:LMPIDdet}
\ee
The following proposition provides the expression of $\tLMPID$.
\bigskip
\begin{prop}
\label{prop:propLMPIDdet}
The LMPID is the following decision rule
\be
\tLMPID=\frac{K-(N-t)+1}{r}p_2(1-p_1)-p_1p_2\test \eta.
\ee
\end{prop}
\begin{proof}
See Appendix \ref{appendix:ProofOfLMPIDdet}.
\end{proof}

As concluding part of this section, we provide the expression of the CLRT, which is the LRT derived under the assumption 
that the ancillary part of the maximal invariant is assigned \cite{lehmann}. The main idea behind
this approach relies on the fact that the ancillary part does not carry any information on the parameter and, hence, can be considered given.
Moreover, it would be of interest to ascertain the existence of a CUMPID, since the UMPID does not exist for the problem at hand when $m<N$.

The test statistic of the CMPID, $\tCLRT$ say, has the same expression as the MPID (see equation (\ref{eq:MPID}))
with the difference that in this case $p_2$ is a constant. Now, observe that $\tCLRT$ is a nonincreasing function of $p_1$. As a matter of fact,
let us evaluate those values of $p_1\in(0,1]$ such that the first derivative of $\tCLRT$ with respect to $p_1$ is less than zero, namely
\begin{align}
\frac{\delta \tCLRT}{\delta p_1}&=-\mbox{SINR} p_2 e^{-\mbox{\scriptsize SINR}  \ p_1 p_2}\nonumber
\\
&\times \sum_{k=0}^{K-(N-t)+1}\binom{K-(N-t)+1}{k}\frac{(r-1)!}{(r+k-1)!}
\nonumber
\\
&\times [\mbox{SINR} \  p_2 (1-p_1)]^k -\mbox{SINR}p_2 e^{- \mbox{\scriptsize SINR}  \ p_1 p_2}\nonumber
\\
&\times \sum_{k=1}^{K-(N-t)+1}\binom{K-(N-t)+1}{k}\frac{(r-1)!}{(r+k-1)!} k 
\nonumber
\\
&\times [\mbox{SINR} \  p_2 (1-p_1)]^{k-1}\leq 0.
\end{align}
The above inequality holds true $\forall p_1\in(0,1]$ and, hence, $\tCLRT$ is a nonincreasing function of $p_1$.
As a consequence of the Karlin-Rubin Theorem, the test
\be
p_1\testINV \eta
\label{eq:CUMPI}
\ee
is UMPI conditionally to $p_2$. Moreover, we can consider the following equivalent form for test (\ref{eq:CUMPI})
\be
\tGLRTb=\frac{1-p_1}{p_1}\test \eta,
\label{eq:CUMPI_equivalent}
\ee
where, as we will show in the next subsection, $\tGLRTb$ coincides with the GLRT in the original data space.

\subsubsection{Case m=N}
Recall that when $m=N$ the maximal invariant statistic obeys the complex beta distribution with $K-r+1,r$ complex degrees of freedom and, hence,
the LRT is given by
\begin{align}
\tLRT &= 
\left.\frac{f_{p_3}(x;H_1)}{f_{p_3}(x;H_0)}\right|_{x=p_3}
=e^{-\mbox{\scriptsize SINR}  \ p_3}
\sum_{k=0}^{K-r+1}\binom{K-r+1}{k}\nonumber
\\
&\times \frac{(r-1)!}{(r+k-1)!}[\mbox{SINR} \  (1-p_3)]^k\test\eta.
\label{eq:UMPID}
\end{align}
Now, following the same steps used to prove that the GLRT is conditionally UMPI, it is easy to show that the test
\be
p_3 \testinv \eta
\ee
is UMPI and statistically equivalent to the Energy Detector (ED) given by
\be
\tED=\frac{1-p_3}{p_3}=\bz_2^\dag\bS_{22}^{-1}\bz_2\test\eta.
\ee
Finally, the following corollary provides the expression of the LMPID when $m=N$.
\begin{corol}
The expression of the LMPID in the case $m=N$ is
\be
\frac{K-r+1}{r}(1-p_3)-p_3\test\eta.
\ee
\end{corol}
\begin{proof}
The above result can be easily obtained following the line of reasoning of the proof of {\em Proposition \ref{prop:propLMPIDdet}}.
\end{proof}
It is clear that the LMPID for $m=N$ is statistically equivalent to the ED.

\subsection{Receivers Obtained as Functions of the Maximal Invariant}
Several invariant decision rules can be obtained exploiting different 
rationales based on either solid theoretical paradigms or heuristic criteria.
The basic observation is the property that
any invariant test can be written as a function of the maximal invariant. Thus,
it is possible to device invariant architectures choosing proper functions of the maximal invariant.
Specifically, we focus on the GLRT, which under some mild technical conditions always leads to an invariant detector \cite{KaySPL},
and the so-called two-step GLRT (2S-GLRT) \cite{Kelly-Nitzberg} which consists in devising
a non-adaptive GLRT for known covariance matrix (step one) and, then, in replacing the unknown covariance matrix 
with a proper estimate (step two) to come with a fully adaptive decision scheme.
Again, in the sequel, we distinguish between the cases $m<N$ and $m=N$.

\subsubsection{Case $m<N$}
Let us begin by providing the test statistic of the GLRT, which has the following expression \cite{ABGR-Subspace,BBORS-Direction}
\be
\tGLRT=\frac{1+\bz^\dag\bS^{-1/2} \bP_{S^{-1/2}E_t}^\perp \bS^{-1/2}\bz}
{1+\bz^\dag\bS^{-1/2} \bP_{S^{-1/2}E_{m}}^\perp \bS^{-1/2}\bz},
\ee
where $\bP_A=\bA(\bA^\dag\bA)^{-1}\bA^\dag$ is the projection matrix onto the subspace spanned by the columns of $\bA\in\C^{N\times M}$,
$\bP^\perp_A=\bI_N-\bP_A$, and $\bE_m=[\bE_t \ \bE_r]$. It is tedious but not difficult to show that
\begin{align}
\bz^\dag\bS^{-1/2} \bP_{S^{-1/2}E_t}^\perp \bS^{-1/2}\bz &= \frac{1-p_1p_2}{p_1p_2},
\label{eq:zSPEtSz}
\\
\bz^\dag\bS^{-1/2} \bP_{S^{-1/2}E_{m}}^\perp \bS^{-1/2}\bz &=\frac{1-p_2}{p_2}.
\label{eq:zSPEmSz}
\end{align}
Sumarizing the statistic of the GLRT can be recast as
\be
\tGLRT=\frac{1}{p_1},
\ee
which is clearly statistically equivalent to the left-hand side of (\ref{eq:CUMPI_equivalent}).

On the other hand, the two-step GLRT (2S-GLRT) is given by
\begin{align}
\tTSGLRT&=\bz^\dag\bS^{-1/2} \bP_{S^{-1/2}E_t}^\perp \bS^{-1/2}\bz \nonumber
\\
&- \bz^\dag\bS^{-1/2} \bP_{S^{-1/2}E_{m}}^\perp \bS^{-1/2}\bz =
\frac{1-p_1}{p_1p_2}.
\end{align}
It follows that both receivers are invariant with respect to $\cL$ and, hence, they ensure the CFAR property with respect
to the unknown covariance matrix of the disturbance and $\bq$.

\subsubsection{Case $m=N$}
In this case, the GLRT and the 2S-GLRT share the same test statistic which has the following expression 
\be
\tGLRT=\tTSGLRT=\bz^\dag\bS^{-1/2} \bP_{S^{-1/2}E_t}^\perp \bS^{-1/2}\bz.
\ee
It is not difficult to show that
\be
\bz^\dag\bS^{-1/2} \bP_{S^{-1/2}E_t}^\perp \bS^{-1/2}\bz = \bz_2^\dag\bS_{22}^{-1}\bz_2=\frac{1-p_3}{p_3},
\ee
which is the statistic of the ED.
It follows that, under the assumption $m=N$, the GLRT, the 2S-GLRT, and the LMPID coincide with the ED, which, in turn,
is UMPI.


\subsection{$P_{fa}$ and $P_d$ of the GLRT, the 2S-GLRT, the LMPID, and ED}
This subsection is aimed at the derivation of closed-form expressions 
for the $P_{fa}$ and the Probability of detection ($P_d$) of the GLRT, the 2S-GLRT, and the LMPID.
To this end, we deal with equivalent decision statistics and distinguish between $m<N$ and $m=N$.

\subsubsection{Case $m<N$}
Let us focus on the GLRT and observe that under the $H_0$ hypothesis, $\tGLRTb\sim\cC\cF_{r,K-(N-t)+1}(0)$, while
under $H_1$, $\tGLRTb\sim\cC\cF_{r,K-(N-t)+1}(\delta)$ with $\delta^2=p_2\mbox{SINR}$. 
It follows that the $P_{fa}$ can be expressed as
\begin{align}
P_{fa}^{\mbox{\tiny GLRT}}(\eta)&=1-F(\eta;r,K-(N-t)+1,0)\nonumber
\\
&=\frac{1}{(1+\eta)^{r+K-N+t}}\sum_{l=0}^{r-1}\binom{r+K-N+t}{l}\eta^l,
\end{align}
where $F(x; n, m, \delta)$, $x\geq 0$, is the Cumulative Distribution Function (CDF) of a random variable ruled by the complex central 
(noncentral) F-distribution with $n,m$ complex degrees of freedom if $\delta=0$ ($\delta>0$). 
Following the same line of reasoning as for the $P_{fa}$, it is not difficult to show that $P_d$ is given by
\begin{align}
P_d^{\mbox{\tiny GLRT}}&(\eta, \mbox{SINR}) = 1-\int_{0}^1 F(\eta;r, K-(N-t)+1, \delta(u)) \nonumber
\\
& \times f_\beta(u;K-(N-t)+r+1,N-t-r,0) du,
\end{align}
where $\delta^2(u)=u\mbox{SINR}$.

As to the 2S-GLRT, observe that it can be recast as $\tTSGLRT=\tGLRTb/p_2$, and, hence, $P_{fa}$ and $P_d$ can be easily
derived relying on the above results, more precisely
\begin{align}
P_{fa}^{\mbox{\tiny 2S-GLRT}}(\eta)&=1-\int_{0}^1 F(\eta u;r, K-(N-t)+1, 0) \nonumber
\\
& \times f_\beta(u;K-(N-t)+r+1,N-t-r,0) du,
\end{align}
\begin{multline}
P_{d}^{\mbox{\tiny 2S-GLRT}}(\eta,\mbox{SINR}) =1
-\int_{0}^1 F(\eta u;r, K-(N-t)+1, \delta(u)) 
\\
\times f_\beta(u;K-(N-t)+r+1,N-t-r,0) du.
\end{multline}
In order to provide the expressions of $P_{fa}$ and $P_d$ for the LMPID, let $a=[K-(N-t)+1]/r$ and observe that
\begin{align}
\tLMPID&=a p_2(1-p_1)- p_1p_2=p_1p_2\left[
a\left(\frac{1}{p_1}-1\right)-1
\right]\nonumber
\\
&=p_1p_2\left[
a\tGLRTb-1
\right]=\frac{p_2}{\tGLRTb+1}(a\tGLRTb-1).
\end{align}
This implies that $P_{fa}$ can be written as\footnote{Notice that $-1\leq \tLMPID\leq a$.}
\begin{align}
&P_{fa}^{\mbox{\tiny LMPID}}(\eta)
\\
&=\mbox{P}\left[ \frac{p_2}{\tGLRTb+1}(a\tGLRTb-1) > \eta; H_0 \right]\nonumber
\\
&=\mbox{P}\left[ \tGLRTb(a p_2-\eta) > p_2+\eta; H_0 \right]\nonumber
\\
&=\int_{0}^1 \mbox{P}\left[ \tGLRTb(a u-\eta) > u+\eta | p_2 = u; H_0 \right]\nonumber
\\
&\times f_\beta(u;K-(N-t)+r+1,N-t-r,0) du\nonumber
\\
&=\left\{
\begin{array}{l}
\ds\int_{\eta/a}^1 \mbox{P}\left[ \tGLRTb > \frac{u+\eta}{a u-\eta} | p_2 = u; H_0 \right]
\\
\times f_\beta(u;K-(N-t)+r+1,N-t-r,0) du
\\
\ds+\int_{0}^{\eta/a} \mbox{P}\left[ \tGLRTb < \frac{u+\eta}{a u-\eta} | p_2 = u; H_0 \right]
\\
\times f_\beta(u;K-(N-t)+r+1,N-t-r,0) du,
\\
\mbox{if} \ 0<\eta\leq a,
\\
0, \quad \mbox{if} \ \eta>a,
\\
\ds\int_{0}^1 \mbox{P}\left[ \tGLRTb > \frac{u+\eta}{a u-\eta} | p_2 = u; H_0 \right]
\\
\times f_\beta(u;K-(N-t)+r+1,N-t-r,0) du,
\\
\mbox{if} \ -1\leq\eta\leq 0,
\\
1, \quad \mbox{if} \ \eta <-1,
\end{array}
\right.\nonumber
\\
&=\left\{
\begin{array}{l}
\ds\int_{\eta/a}^1 \left\{1-F\left(\frac{u +\eta}{u a-\eta};r, K-(N-t)+1, 0 \right) \right\}
\\
\times f_\beta(u;K-(N-t)+r+1,N-t-r,0) du
\\
\ds+\int_{0}^{\eta/a} F\left(\frac{u +\eta}{u a-\eta};r, K-(N-t)+1, 0 \right)
\\
\times f_\beta(u;K-(N-t)+r+1,N-t-r,0) du, 
\\
\mbox{if} \ 0<\eta\leq a,
\\
0,  \quad \mbox{if} \ \eta>a
\\
\ds 1-\int_{0}^1 F\left(\frac{u +\eta}{u a-\eta};r, K-(N-t)+1, 0 \right) 
\\
\times f_\beta(u;K-(N-t)+r+1,N-t-r,0) du, 
\\
\mbox{if} \ -1\leq \eta\leq 0,
\\
1, \quad \mbox{if} \ \eta<-1.
\end{array}
\right.
\end{align}
The expression of the $P_d$ is formally analogous to that of the $P_{fa}$ but for the presence of the noncentrality 
parameter in the distribution of $\tGLRTb$, namely
\begin{align}
&P^{\mbox{\tiny LMPID}}_d(\eta,\mbox{SINR})\nonumber
\\
&=\left\{
\begin{array}{l}
\ds\int_{\eta/a}^1 \left\{1-F\left(\frac{u +\eta}{u a-\eta};r, K-(N-t)+1, \delta(u) \right) \right\}
\\
\times f_\beta(u;K-(N-t)+r+1,N-t-r,0) du
\\
\ds+\int_{0}^{\eta/a} F\left(\frac{u +\eta}{u a-\eta};r, K-(N-t)+1, \delta(u) \right)
\\
\times f_\beta(u;K-(N-t)+r+1,N-t-r,0) du, 
\\
\mbox{if} \ 0<\eta\leq a,
\\
0, \quad \mbox{if} \ \eta > a,
\\
\ds 1-\int_{0}^1 F\left(\frac{u +\eta}{u a-\eta};r, K-(N-t)+1, \delta(u) \right) 
\\
\times f_\beta(u;K-(N-t)+r+1,N-t-r,0) du, 
\\
\mbox{if} \ -1\leq \eta\leq 0,
\\
1, \quad \mbox{if} \ \eta < -1.
\end{array}
\right.
\end{align}

\subsubsection{Case $m=N$}
Under the assumption $m=N$, the GLRT, the 2S-GLRT, and the LMPID are statistically equivalent to the ED. For this reason,
the object of this section is the statistical characterization of the ED. Specifically, 
it is easy to show that $\tED$ is ruled by
\begin{itemize}
\item the complex central F-distribution with $r,K-r+1$ complex degrees of freedom under $H_0$;
\item the complex noncentral F-distribution with $r,K-r+1,$ degrees of freedom and noncentrality 
parameter $\delta$ with $\delta^2=\theta_2^\dag\bM_{22}^{-1}\theta_2$ under $H_1$.
\end{itemize}
It follows that
\be
P_{fa}^{\mbox{\tiny ED}}(\eta)=\frac{1}{(1+\eta)^{K}}\sum_{l=0}^{r-1}\binom{K}{l}\eta^l
\ee
and
\be
P_{d}^{\mbox{\tiny ED}}(\eta,\mbox{SINR})=1- F(\eta;r, K-r+1, \delta).
\ee

\subsection{Illustrative Examples}
Figures \ref{fig:fig01}-\ref{fig:fig04} display $P_d$ versus the SINR
for the considered decision schemes.
The curves have been obtained by means the above closed-form formulas which have been computed by standard numerical routines but for
the MPID, whose performance has been obtained via standard Monte Carlo counting techniques. 
More precisely, the thresholds necessary to ensure a preassigned value of
$P_{fa}$ have been evaluated exploiting $100/P_{fa}$ independent trials, while the $P_d$
values are estimated over $5000$ independent trials.
As to the disturbance, it is modeled as an
exponentially-correlated Gaussian vector with covariance matrix $\bM = \sigma^2_n\bI_N + \sigma^2_c\bM_c$,
where $\sigma^2_n>0$ is the thermal noise power, $\sigma^2_c>0$ is the clutter power, and
the $(i,j)$-th element of $\bM_c$ is given by $0.95^{|i-j|}$. 
The clutter-to-noise ratio $\sigma_c^2/\sigma_n^2$ and the interferer-to-noise 
ratio $\|\theta_1\|^2/\sigma_n^2$ are both set to 30 dB with $\sigma^2_n=1$.
Finally, all numerical examples assume $P_{fa}=10^{-4}$.

Figures \ref{fig:fig01} and \ref{fig:fig02} show the performances of the considered detectors assuming $N=8$, $K=12$, and
different values of $(r,t)$. Inspection of the figures highlights that the GLRT, the 2S-GLRT, and the MPID share
practically the same performance, while the LMPID performs poorer. In particular, the latter exhibits a loss of 
about 3 dB at $P_d=0.9$ in Figure \ref{fig:fig01} and does not achieve $P_d=0.9$ for the SNR values considered in Figure \ref{fig:fig02}.
Moreover, the GLRT and the MPID are slightly superior to the 2S-GLRT. 
Finally, it is worth observing that for high values  of $r$,
the $P_d$ curves of the considered detectors move toward high values of SINR, namely 
there is a performance degradation with respect to the case where $r$ is low; 
this is strictly tied to the increase of the number of unknown parameters to be estimated.
In Figures \ref{fig:fig03} and \ref{fig:fig04}, $P_d$ versus SINR is plotted assuming $N=8$ and $K=16$. Inspection of the figures highlights that
the above hierarchy keeps unaltered, while the increased number of secondary data used for estimation purposes decreases 
the loss of the 2S-GLRT with respect to the GLRT. Finally, notice that contrary to the previous case the $P_d$ of the LMPID
can be higher than $0.9$ for SINR values greater than 16 dB when $r=2$ and 18 dB when $r=4$.

\section{Conclusions}
In this paper we have considered adaptive radar detection of a  
subspace signal embedded in two sources of interference: a Gaussian  
component with unknown covariance matrix modeling the joint effect of  
clutter plus receiver noise and a subspace structured part accounting  
for the effect of coherent jammers impinging on the radar antenna. We  
have formulated the problem as a binary hypothesis test and have  
exploited the theory of invariance to characterize adaptive detectors  
enjoying some relevant symmetries. In this context the main achieved  
technical results are:
\begin{itemize}
\item
The identification of a suitable group of transformations leaving the  
considered hypothesis test invariant and the interpretation of the  
group action as a mean to force at the design stage the CFAR property  
with respect to clutter plus noise covariance matrix and jammer  
parameters.
\item
The design and the statistical characterization of a maximal invariant  
statistic which organizes the radar data in equivalence classes and  
significantly realizes a compression of the original observation domain.
\item
The synthesis of the optimum invariant detector and the discussion on  
the existence of the UMPI test. In this respect, we have shown that  
the UMPI detector in general does not exists but for the case where  
the useful signal plus jammer subspace size completely fills the  
dimension of the observation domain.
\item
The synthesis of the LMPID, the CUMPID, and the proof that the latter  
test coincides with the GLRT.
\item
The development of analytic expressions for the performance assessment of some  
practically implementable invariant detectors.
\end{itemize}
Possible future research tracks might concern the possibility to deal  
with a partially homogeneous Gaussian interference through the  
additional invariance under a common scaling of the secondary data, as  
well as the eventuality to force some structure in the covariance  
matrix of the Gaussian interference (as for instance persymmetry \cite{hongbinPersymmetric,Pascal}).  
Last but not least, it would be interesting the extension of the  
entire framework to a non-homogeneous scenario where the training data  
may exhibit different power levels.

\appendices

\section{Invariance of Problem (\ref{eq:hypothesistest01}) with respect to the Group $\cL$}
\label{appendix:InvariantTransformations}

Let $(\bG,\bff)\in\cL$ and observe that, under $H_1$, $\bx=\bG\bz+\bff$ obeys the complex normal distribution with covariance matrix
$\bG\bM\bG^\dag$ and mean
\begin{align}
&\bG(\bE_t\btheta_{11}+\bE_r\btheta_{2}) + \bff=
\left[
\begin{array}{c}
\bG_{11} \btheta_{11}
\\
\bzero
\\
\bzero
\end{array}
\right]
+
\left[
\begin{array}{c}
\bG_{12} \btheta_{2}
\\
\bG_{22} \btheta_{2}
\\
\bzero
\end{array}
\right]\nonumber
\\
&+
\left[
\begin{array}{c}
\bff_{11}
\\
\bzero
\\
\bzero
\end{array}
\right]
=\bE_t\bar{\btheta}_{1}+\bE_r\bar{\btheta}_{2},
\end{align}
where $\bar{\btheta}_{1}=\bG_{11} \btheta_{11} + \bG_{12} \btheta_{2}+\bff_{11}$ and $\bar{\btheta}_{2}=\bG_{22} \btheta_{2}$. Moreover, $\bx$ is independent
of $\bx_k=\bG\bz_k$, $k=1,\ldots,K$, which are iid complex normal vectors with zero mean and covariance matrix $\bG\bM\bG^\dag$. Now, observe that
$\| \bar{\btheta}_{2} \|>0$ if and only if $\|\btheta_{2}\|>0$ (since $\bG_{22}$ is nonsingular).

On the other hand, when $H_0$ is in force, $\bx$ is distributed as under $H_1$ but for the mean, which is given by
\be
\bG\bE_t\btheta_{10}+\bff=
\left[
\begin{array}{c}
\bG_{11}\btheta_{10}+\bff_{11}
\\
\bzero
\\
\bzero
\end{array}
\right]=\bE_t\bar{\btheta}_{10},
\ee
where $\bar{\btheta}_{10}=\bG_{11}\btheta_{10}+\bff_{11}$, whereas $\bx_k$, $k=1,\ldots,K$, are still iid complex normal random vectors with zero mean and
covariance matrix $\bG\bM\bG^\dag$.

Thus, we have shown that the original partition of the parameter space, data distribution, and the structure of the 
subspace containing the useful signal components are preserved after the trasformation $(\bG,\bff)$.

\section{Derivation of the Maximal Invariant Statistics}
\label{appendix:maximalInvariantsProof}

This appendix is devoted to the proof of {\em Proposition \ref{proposition:derivationMaxInv}}.
In particular, we focus on the more challenging case $2\leq m <N$ since $m=N$ can be obtained according
to the same proof method.

Before proceeding, it is worth recalling that a statistic $\bt(\bz, \bS)$ is said to be
maximal invariant statistic with respect to a group of transformations
$\cL$ if and only if \cite{lehmann,Scharf-book}
\be
\bt(\bz, \bS)=\bt[l(\bz, \bS)], \quad \forall l\in\cL,
\label{eq:maxInv1Prop}
\ee
and
\be
\bt(\bz, \bS)=\bt(\bar{\bz}, \bar{\bS})  \\
\Rightarrow \exists \ l\in\cL \ : \  (\bz, \bS)=l(\bar{\bz}, \bar{\bS}).
\label{eq:maxInv2Prop}
\ee
In order to prove (\ref{eq:maxInv1Prop}), we consider the following vector
\be
\left[
\begin{array}{c}
\bz_{23}^\dag\bS_2^{-1}\bz_{23}
\\
\bz_3^\dag\bS_{33}^{-1}\bz_3
\end{array}
\right]
\label{eq:maxInv_01}
\ee
where
\be
\bz_{23}=\left[
\begin{array}{c}
\bz_2
\\
\bz_3
\end{array}
\right], \quad
\bS_{2}=\left[
\begin{array}{cc}
\bS_{22} & \bS_{23}
\\
\bS_{32} & \bS_{33}
\end{array}
\right],
\ee
and 
\be
\bz_{23}^\dag\bS_2^{-1}\bz_{23}=\bz_3^\dag\bS_{33}^{-1}\bz_3+\bz_{2.3}\bS_{2.3}^{-1}\bz_{2.3}.
\ee
Moreover, let us partition the matrix $\bG$ as follows
\be
\bG=\left[
\begin{array}{cc}
\bG_1 & \bG_{2}
\\
\bzero & \bG_3
\end{array}
\right],
\ee
where $\bG_1=\bG_{11}\in\C^{t\times t}$, $\bG_{2}=[\bG_{12} \ \bG_{13}]\in\C^{t\times (N-t)}$,  and 
\be
\bG_{3}=\left[
\begin{array}{cc}
\bG_{22} & \bG_{23}
\\
\bzero & \bG_{33}
\end{array}
\right]\in\C^{(N-t)\times (N-t)}.
\label{eq:blockG3}
\ee
Now, let $(\bar{\bz},\bar{\bS})=l(\bz,\bS)$ and observe that\footnote{We partition $\bar{\bz}$ and $\bar{\bS}$ according to the same rule used for $\bz$ and $\bS$. For this reason in the sequel we omit the definition of the submatrices and subvectors of $\bar{\bz}$ and $\bar{\bS}$.}
\be
\bar{\bz}=\bG\bz + \bff = \left[
\begin{array}{c}
\bG_1\bz_1+\bG_{2}\bz_{23}+\bff_{11}
\\
\bG_{3}\bz_{23}
\end{array}
\right]
\ee
and
\be
\bar{\bS}=\bG\bS\bG^\dag=\left[
\begin{array}{cc}
\bA_1 & \bA_2
\\
\bA_3 & \bG_3\bS_2\bG_3^\dag
\end{array}
\right],
\label{eq:GSG_GSbG}
\ee
where $\bA_i$, $i=1,2,3$, are matrices of proper size and
\be
\bG_3\bS_2\bG_3^\dag=\left[
\begin{array}{cc}
\bB_1 & \bB_2
\\
\bB_3 & \bG_{33}\bS_{33}\bG_{33}^\dag
\end{array}
\right]
\ee
with, in turn, $\bB_i$, $i=1,2,3$ matrices of proper dimensions.
It follows that
\begin{align}
&\left[
\begin{array}{c}
\bar{\bz}_{23}^\dag\bar{\bS}_2^{-1}\bar{\bz}_{23}
\\
\bar{\bz}_3^\dag\bar{\bS}_{33}^{-1}\bar{\bz}_3
\end{array}
\right]
=
\left[
\begin{array}{c}
\bz_{23}^\dag\bG_3^\dag (\bG_3^\dag)^{-1}\bS_2^{-1}\bG_3^{-1}\bG_3\bz_{23}
\\
\bz_3^\dag\bG_{33}(\bG_{33}^\dag)^{-1}\bS_{33}^{-1}\bG_{33}^{-1}\bG_{33}\bz_3
\end{array}
\right]\nonumber
\\
&=
\left[
\begin{array}{c}
\bz_{23}^\dag\bS_2^{-1}\bz_{23}
\\
\bz_3^\dag\bS_{33}^{-1}\bz_3
\end{array}
\right]
\end{align}
and, hence, that
\be
\bz_{2.3}\bS_{2.3}^{-1}\bz_{2.3}=\bar{\bz}_{2.3}\bar{\bS}_{2.3}^{-1}\bar{\bz}_{2.3}.
\ee

The proof of the second property (\ref{eq:maxInv2Prop}) adopts the following rationale: first, we find the submatrix $\bG_3$ and then we construct
$\bff_{11}$ and the remaining blocks of $\bG$. 

Assume that there exist $(\bz,\bS)$ and $(\bar{\bz},\bar{\bS})$ such that
\be
\bt(\bz,\bS)=\left[
\begin{array}{c}
\bz_{2.3}^\dag\bS_{2.3}^{-1}\bz_{2.3}
\\
\bz_{3}^\dag\bS^{-1}_{33}\bz_3
\end{array}
\right]=
\left[
\begin{array}{c}
\bar{\bz}_{2.3}^\dag\bar{\bS}_{2.3}^{-1}\bar{\bz}_{2.3}
\\
\bar{\bz}_{3}^\dag\bar{\bS}^{-1}_{33}\bar{\bz}_3
\end{array}
\right]=\bt(\bar{\bz},\bar{\bS}).
\ee
The above equalities can be recast using the Euclidean norm of a vector, more precisely
\begin{align}
\| \by_{2.3} \|^2 &= \| \bar{\by}_{2.3} \|^2,
\\
\| \by_{3} \|^2 &= \| \bar{\by}_{3} \|^2,
\end{align}
where $\by_{2.3}=\bS_{2.3}^{-1/2}\bz_{2.3}$, $\bar{\by}_{2.3}=\bar{\bS}_{2.3}^{-1/2}\bar{\bz}_{2.3}$, 
$\by_{3}=\bS_{33}^{-1/2}\bz_{3}$, and $\bar{\by}_{3}=\bar{\bS}_{33}^{-1/2}\bar{\bz}_{3}$. As a consequence, there exist $\bU_3\in\C^{(N-m)\times (N-m)}$ and
$\bU_{2.3}\in\C^{r\times r}$ unitary matrices such that
\be
\by_{2.3}=\bU_{2.3}\bar{\by}_{2.3} \quad \mbox{and} \quad \by_3=\bU_3\bar{\by}_{3}.
\ee
Moreover, partition $\bS_2^{-1}$ and $\bar{\bS}_2^{-1}$ as follows
\begin{align}
\bS_2^{-1} &=
\underbrace{
\left[
\begin{array}{cc}
\bI_r & \bzero
\\
-\bS_{33}^{-1}\bS_{32} & \bI_{N-m}
\end{array}
\right]}_{\bP^\dag}
\underbrace{\left[
\begin{array}{cc}
\bS_{2.3}^{-1} & \bzero
\\
\bzero & \bS_{33}^{-1}
\end{array}
\right]}_{\bW^2}\nonumber
\\
&\times \underbrace{\left[
\begin{array}{cc}
\bI_r & -\bS_{23}\bS_{33}^{-1}
\\
\bzero & \bI_{N-m}
\end{array}
\right],}_{\bP}
\label{eq:partitionS2}
\\
\bar{\bS}_2^{-1} &=
\underbrace{
\left[
\begin{array}{cc}
\bI_r & \bzero
\\
-\bar{\bS}_{33}^{-1}\bar{\bS}_{32} & \bI_{N-m}
\end{array}
\right]}_{\bar{\bP}^\dag}
\underbrace{\left[
\begin{array}{cc}
\bar{\bS}_{2.3}^{-1} & \bzero
\\
\bzero & \bar{\bS}_{33}^{-1}
\end{array}
\right]}_{\bar{\bW}^2}\nonumber
\\
&\times \underbrace{\left[
\begin{array}{cc}
\bI_r & -\bar{\bS}_{23}\bar{\bS}_{33}^{-1}
\\
\bzero & \bI_{N-m}
\end{array}
\right],}_{\bar{\bP}}
\label{eq:partitionS2bar}
\end{align}
and note that
\begin{align}
\bW\bP\bz_{23}&=\left[
\begin{array}{c}
\bS_{2.3}^{-1/2}\bz_{2.3}
\\
\bS_{33}^{-1/2}\bz_3
\end{array}
\right]=
\left[
\begin{array}{c}
\by_{2.3}
\\
\by_3
\end{array}
\right],
\\
\bar{\bW}\bar{\bP}\bz_{23}&=\left[
\begin{array}{c}
\bar{\bS}_{2.3}^{-1/2}\bar{\bz}_{2.3}
\\
\bar{\bS}_{33}^{-1/2}\bar{\bz}_3
\end{array}
\right]=
\left[
\begin{array}{c}
\bar{\by}_{2.3}
\\
\bar{\by}_3
\end{array}
\right].
\end{align}
Gathering the above results yields
\begin{align}
&\bW\bP\bz_{23}=\underbrace{\left[
\begin{array}{cc}
\bU_{2.3} & \bzero
\\
\bzero & \bU_3
\end{array}
\right]}_{\bU_1}\bar{\bW}\bar{\bP}\bar{\bz}_{23}
\\
&\Rightarrow
\bz_{23}=(\bW\bP)^{-1}\bU_1\bar{\bW}\bar{\bP}\bar{\bz}_{23},
\label{eq:z23eqz23bar}
\end{align}
where $\bD_3=(\bW\bP)^{-1}\bU_1\bar{\bW}\bar{\bP}$ is an upper block-triangular matrix with the same structure as $\bG_3$ in (\ref{eq:blockG3}).
Besides, from equations (\ref{eq:partitionS2}) and (\ref{eq:partitionS2bar}) it follows that
\begin{align}
& \bW\bP\bS_2\bP^\dag\bW=\bI_{N-m}=\bU_1\bU_1^\dag
\\
&\Rightarrow \bU_1^\dag\bW\bP\bS_2(\bU_1^\dag\bW\bP)^\dag=\bI_{N-m}=\bar{\bW}\bar{\bP}\bar{\bS_2}(\bar{\bW}\bar{\bP})^\dag
\\
& \Rightarrow \bS_2=(\bW\bP)^{-1}\bU_1\bar{\bW}\bar{\bP}\bar{\bS_2}(\bar{\bW}\bar{\bP})^\dag\bU_1^\dag[(\bW\bP)^\dag]^{-1}
\\
&=\bD_3\bar{\bS}_2\bD_3^\dag.
\label{eq:S3G3S3b}
\end{align}
So far, we have constructed the block $\bG_3=\bD_3$ of $\bG$, it still remains to find the other blocks of $\bG$. To this end, 
let $\bS_3=[\bS_{12} \ \bS_{13}]$, $\bS_1=\bS_{11}$ and write (\ref{eq:GSG_GSbG}) as
\begin{align}
&\bG^{-1}\bS=\bar{\bS}\bG^\dag
\\
&\Rightarrow 
\left[  
\begin{array}{cc}
\bG_1^{-1} & -\bG_1^{-1}\bG_{2}\bG_3^{-1}
\\
\bzero & \bG_3^{-1}
\end{array}
\right]\nonumber
\\
& \times
\left[  
\begin{array}{cc}
\bS_{1} & \bS_{3}
\\
\bS_{3}^\dag & \bS_{2}
\end{array}
\right]
=\left[  
\begin{array}{cc}
\bar{\bS}_{1} & \bar{\bS}_{3}
\\
\bar{\bS}_{3}^\dag & \bar{\bS}_{2}
\end{array}
\right]
\left[  
\begin{array}{cc}
\bG_1^\dag & \bzero
\\
\bG_{2}^\dag & \bG_{3}^\dag
\end{array}
\right]
\\
&\Rightarrow 
\left[  
\begin{array}{cc}
\bG_1^{-1}\bS_{1}-\bG_{1}^{-1}\bG_{2}\bG_{3}^{-1}\bS_{3}^\dag & \bG_1^{-1}\bS_{3}-\bG_1^{-1}\bG_{2}\bG_3^{-1}\bS_{2}
\\
\bG_3^{-1}\bS_3^\dag & \bG_3^{-1}\bS_2
\end{array}
\right]\nonumber
\\
&=\left[  
\begin{array}{cc}
\bar{\bS}_{1}\bG_1^\dag+\bar{\bS}_{3}\bG_2^\dag & \bar{\bS}_{3}\bG_3^\dag
\\
\bar{\bS}_{3}^\dag\bG_1^\dag+\bar{\bS}_2\bG_2^\dag & \bar{\bS}_{2}\bG_3^\dag
\end{array}
\right].
\end{align}
The last equation is equivalent to the following system of equations
\be
\left\{
\begin{array}{l}
\bG_1^{-1}\bS_{1}-\bG_{1}^{-1}\bG_{2}\bG_{3}^{-1}\bS_{3}^\dag=\bar{\bS}_{1}\bG_1^\dag+\bar{\bS}_{3}\bG_2^\dag,
\\
\bG_1^{-1}\bS_{3}-\bG_1^{-1}\bG_{2}\bG_3^{-1}\bS_{2}=\bar{\bS}_{3}\bG_3^\dag,
\\
\bG_3^{-1}\bS_3^\dag = \bar{\bS}_{3}^\dag\bG_1^\dag+\bar{\bS}_2\bG_2^\dag.
\end{array}
\right.
\label{eq:systemEquation}
\ee
Replacing the following equality\footnote{Observe that it comes from (\ref{eq:S3G3S3b}).}
\be
\bG_3^{-1}\bS_2 = \bar{\bS}_{2}\bG_3^\dag
\ee
into the second equation of the above system yields
\begin{align}
& \bG_1^{-1}\bS_{3}-\bG_1^{-1}\bG_{2}\bar{\bS}_{2}\bG_3^\dag=\bar{\bS}_{3}\bG_3^\dag
\\
& \Rightarrow \bS_{3}-\bG_{2}\bar{\bS}_{2}\bG_3^\dag=\bG_1\bar{\bS}_{3}\bG_3^\dag
\\
& \Rightarrow (\bG_3^{-1}\bS_{3}^\dag)^\dag=\bG_1\bar{\bS}_{3}+\bG_{2}\bar{\bS}_{2}
\\
& \Rightarrow \bG_3^{-1}\bS_{3}^\dag=\bar{\bS}_{3}^\dag\bG_1^\dag+\bar{\bS}_{2}\bG_{2}^\dag.
\end{align}
Thus, the second matrix equation of (\ref{eq:systemEquation}) is redundant and can be neglected.
From the third equation it stems that
\be
\bG_2^\dag=\bar{\bS}_2^{-1}\left(\bG_3^{-1}\bS_3^\dag -\bar{\bS}_{3}^\dag\bG_1^\dag\right)
\label{eq:G2}
\ee
which, replaced in the first equation of (\ref{eq:systemEquation}), leads to
\begin{align}
&\bG_1^{-1}\bS_1-\bG_1^{-1}\left[
\bS_{3}(\bG_3^\dag)^{-1}-\bG_1\bar{\bS}_3
\right]\bar{\bS}_2^{-1}\bG_3^{-1}\bS_3^\dag\nonumber
\\
&=
\bar{\bS}_1\bG_1^\dag+\bar{\bS_3}\bar{\bS}_2^{-1}(\bG_3^{-1}\bS_3^\dag-\bar{\bS}_3^\dag\bG_1^\dag)
\\
&\Rightarrow
\bS_1-\bS_{3}\bS_2^{-1}\bS_3^\dag
+\bG_1\bar{\bS}_3\bar{\bS}_2^{-1}\bG_3^{-1}\bS_3^\dag\nonumber
\\
&=
\bG_1\bar{\bS}_1\bG_1^\dag+\bG_1\bar{\bS_3}\bar{\bS}_2^{-1}\bG_3^{-1}\bS_3^\dag-\bG_1\bar{\bS}_3\bar{\bS}_2^{-1}\bar{\bS}_3^\dag\bG_1^\dag
\\
&\Rightarrow
\bG_1( \bar{\bS}_1-\bar{\bS}_3\bar{\bS}_2^{-1}\bar{\bS}_3^\dag )\bG_1^\dag
={\bS}_1-{\bS}_3{\bS}_2^{-1}{\bS}_3^\dag
\end{align}
The solution of the last equation has the following expression
\be
{\bG}_1=({\bS}_1-{\bS}_3{\bS}_2^{-1}{\bS}_3^\dag)^{1/2}( \bar{\bS}_1-\bar{\bS}_3\bar{\bS}_2^{-1}\bar{\bS}_3^\dag )^{-1/2}
\label{eq:G1}
\ee
and it is such that $\det({\bG}_1)\neq 0$. Replacing $\bG_1$ in (\ref{eq:G2}) with (\ref{eq:G1}) yields the expression of $\bG_2$.

Finally, $\bff_{11}$ can be evaluated as follows
\begin{align}
\bff&=\bz-\bG\bar{\bz}=\left[
\begin{array}{c}
\bz_1
\\
\bz_{23}
\end{array}
\right]-
\left[
\begin{array}{cc}
\bG_1 & \bG_2
\\
\bzero & \bG_3
\end{array}
\right]
\left[
\begin{array}{c}
\bar{\bz}_1
\\
\bar{\bz}_{23}
\end{array}
\right]
\\
&=
\left[
\begin{array}{c}
\bz_1-\bG_1\bar{\bz}_1-\bG_2\bar{\bz}_{23}
\\
\bzero
\end{array}
\right],
\end{align}
where the last equality comes from (\ref{eq:z23eqz23bar}).

\section{Statistical Characterization of the Maximal Invariant}
\label{appendix:ProofStatisticalCharacterization}

The joint pdf of $p_1$ and $p_2$ can be computed exploiting the following equality
\be
f_{p_1,p_2}(x,y)=f_{p_1|p_2}(x|p_2=y)f_{p_2}(y),
\ee
where $f_{p_1|p_2}(x|p_2=y)$ is the conditional pdf of $p_1$ given $p_2$ and $f_{p_2}(y)$ is the marginal pdf of $p_2$.

Let us focus on $f_{p_2}(y)$ and observe that from (\ref{eq:distribuzione_w}) it stems that $\bw_3\sim\cC\cN_{N-m}(\bzero,\bI_{N-m})$, while
by {\em Theorem A.11} of \cite{BOR-Morgan} $\bS_{033}\sim\cC\cW_{N-m}(K,\bI_{N-m})$. Now, recast $\bw_3^\dag\bS_{033}^{-1}\bw_3$ as
\be
\bar{p}_2=\ds\frac{\bw_3^\dag\bw_3}
{\ds\frac{\bw_3^\dag\bw_3}
{\bw_3^\dag\bS_{033}^{-1}\bw_3}}=\frac{a}{b}
\ee
where $a\sim\cC\chi^2_{N-m}(0)$ and, by {Theorem A.13} of \cite{BOR-Morgan}, $b\sim\cC\chi^2_{K-(N-m)+1}(0)$ is statistically independent of $a$.
It follows that $\bar{p}_2$ is ruled by the complex central F-distribution with $N-m,K-(N-m)+1$ complex degrees of freedom 
and, hence, $p_2=1/(1+\bar{p}_2)\sim\cC\beta_{K-(N-m)+1,N-m}(0)$.
Note that $p_2$ is an ancillary statistic and its distribution is the one and the same under both hypotheses.

The distribution of $p_1$ can be obtained assuming that the ``3-components'', namely the entries of the considered random vectors
sharing the subscript equal to 3, are no longer random variables but are assigned. In such a case, $\bar{p}_2$ is deterministic
and under $H_1$
\begin{itemize}
\item $\bd=(\bw_2-\bS_{023}\bS_{033}^{-1}\bw_3)/\sqrt{1+\bw_3^\dag\bS_{033}^{-1}\bw_3}
\sim\cC\cN_{r}((\bM_{22}-\bM_{23}\bM_{33}^{-1}\bM_{32})^{-1/2}\btheta_2,\bI_r)$;
\item $\bX=(\bS_{022}-\bS_{023}\bS_{033}^{-1}\bS_{032})\sim\cC\cW_{r}(K-(N-t)+r,\bI_{r})$ independent of $\bd$ 
(see {\em Theorem A.11} of \cite{BOR-Morgan}).
\end{itemize}
Again, exploiting {\em Theorem A.13} of \cite{BOR-Morgan} leads to
\be
\bar{p}_1=\frac{\ds\bd^\dag\bd}
{\ds
\frac{\ds\bd^\dag\bd}{\ds \bd^\dag\bX^{-1}\bd}
}\sim\cC\cF_{r,K-(N-t)+1}(\delta),
\ee
where
\be
\delta^2=\frac{\btheta_2^\dag(\bM_{22}-\bM_{23}\bM_{33}^{-1}\bM_{32})^{-1}\btheta_2}
{1+\bw_3^\dag\bS_{033}^{-1}\bw_3}
\ee
is the noncentrality parameter. As a consequence, given $p_2$, $p_1=1/(1+\bar{p}_1)\sim\cC\beta_{K-(N-t)+1,r}(\delta)$ with $\delta^2=\mbox{SINR}p_2$.
It is clear that under $H_0$ and given $p_2$, $p_1\sim\cC\beta_{K-(N-t)+1,r}(0)$, which does not depend on $p_2$.

Gathering the above results, the joint pdf of $p_1$ and $p_2$ can be written as
\begin{multline}
f_{p_1,p_2}(x,y; H_1)=f_{\beta}(x;K-(N-t)+1,r,0)
\\
\times f_{\beta}(y;K-(N-t)+r+1,N-t-r,0)
\\
\times e^{-i \ \mbox{\scriptsize SINR}  \ x y}\sum_{k=0}^{K-(N-t)+1}\binom{K-(N-t)+1}{k}
\\
\times \frac{(r-1)!}{(r+k-1)!}[\mbox{SINR} \  y (1-x)]^k,
\end{multline}
under $H_1$, and
\begin{align}
f_{p_1,p_2}(x,y; H_0)&=f_{\beta}(x;K-(N-t)+1,r,0)\nonumber
\\
&\times f_{\beta}(y;K-(N-t)+r+1,N-t-r,0),
\end{align}
under $H_0$, where
\be
f_{\beta}(x;n,m,0)=\frac{(n+m+1)!}{(n-1)!(m-1)!}x^{n-1}(1-x)^{m-1}.
\ee

\section{Derivation of the LMPID}
\label{appendix:ProofOfLMPIDdet}
It is not difficult to show that
\begin{align}
&\tLMPID\nonumber
\\
&=\Bigg\{\frac{\delta}{\delta\mbox{SINR}}
\Bigg[
e^{- \mbox{\scriptsize SINR}  \ p_1 p_2} \sum_{k=0}^{K-(N-t)+1}\binom{K-(N-t)+1}{k} \nonumber
\\
&\times\left.\left.\frac{(r-1)!}{(r+k-1)!}[\mbox{SINR} \  p_2 (1-p_1)]^k
\right]
\right\}_{\mbox{\scriptsize SINR=0}}\nonumber
\\
&=\Bigg\{-p_1 p_2 e^{- \mbox{\scriptsize SINR}  \ p_1 p_2}\sum_{k=0}^{K-(N-t)+1}\binom{K-(N-t)+1}{k}\nonumber
\\
&\times \frac{(r-1)!}{(r+k-1)!}[\mbox{SINR} \  p_2 (1-p_1)]^k+e^{-\mbox{\scriptsize SINR}  \ p_1 p_2} 
\nonumber
\\
&\times 
\sum_{k=1}^{K-(N-t)+1}\binom{K-(N-t)+1}{k}\frac{(r-1)!}{(r+k-1)!}\nonumber
\\
&\times k \mbox{SINR}^{k-1} \ [p_2 (1-p_1)]^k\Bigg\}_{\mbox{\scriptsize SINR=0}}\nonumber
\\
&=\frac{K-(N-t)+1}{r}p_2(1-p_1)-p_1p_2.
\end{align}


\setlength{\unitlength}{1cm}
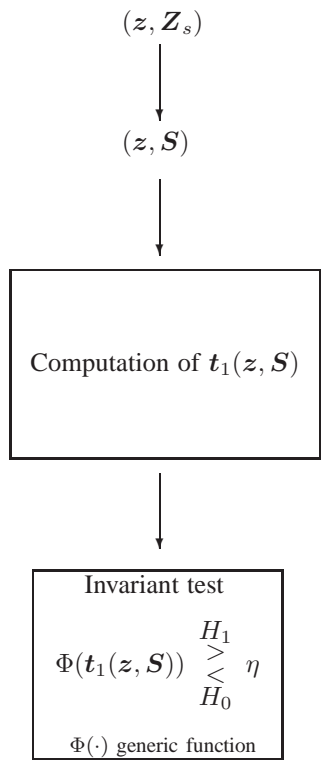
\begin{figure}[ht!]
\begin{center}
\begin{picture}(8,11)(0,0)
\put(4,11){$(\bz,\bZ_s)$}
\put(4.5,10.8){\vector(0,-1){1}}
\put(4,9.4){$(\bz,\bS)$}
\put(4.5,9){\vector(0,-1){1}}
\put(2.4,5.3){ \frame{\makebox(4.1,2.5){ Computation of $\bt_1(\bz,\bS)$ }}}
\put(4.5,5.1){\vector(0,-1){1}}
\put(2.7,1.3){ \frame{\makebox(3.3,2.5){ $\Phi(\bt_1(\bz,\bS))\test\eta$ }}}
\put(3.5,3.5){Invariant test}
\put(3.3,1.4){{\footnotesize $\Phi(\cdot)$ generic function}}
\end{picture}
\caption{Block diagram of the transformations that lead to a generic invariant test.}
\label{fig:trasformationInvariantTest}
\end{center}
\end{figure}

\begin{figure}[!htp]
    \centering
   \includegraphics[width=7cm]{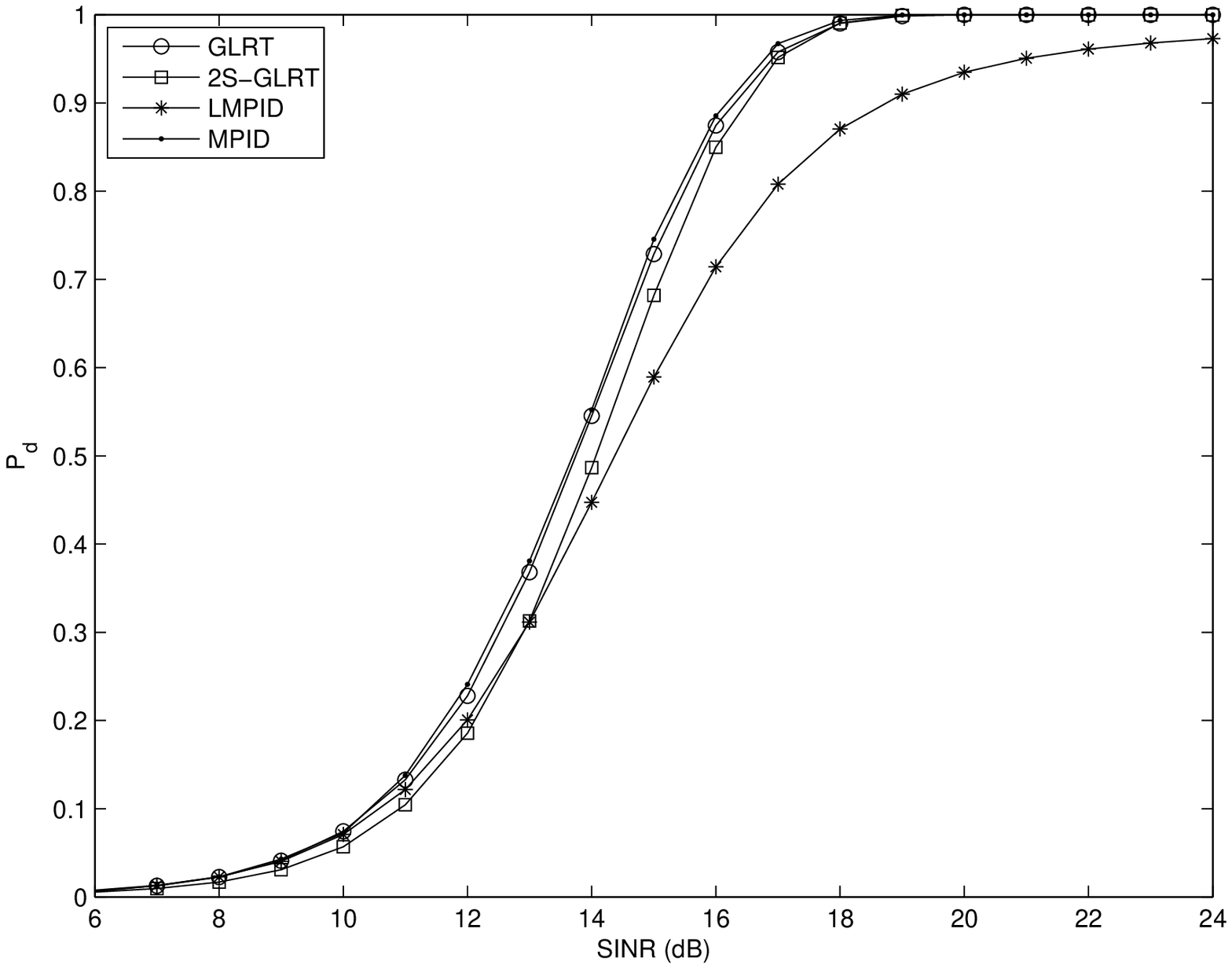}
    \caption{$P_d$ versus SINR for the GLRT, the 2S-GLRT, the LMPID, and the MPID assuming $N=8$, $K=12$, $r=2$, $t=4$, and $P_{fa}=10^{-4}$.}
    \label{fig:fig01}
\end{figure}

\begin{figure}[!htp]
    \centering
   \includegraphics[width=7cm]{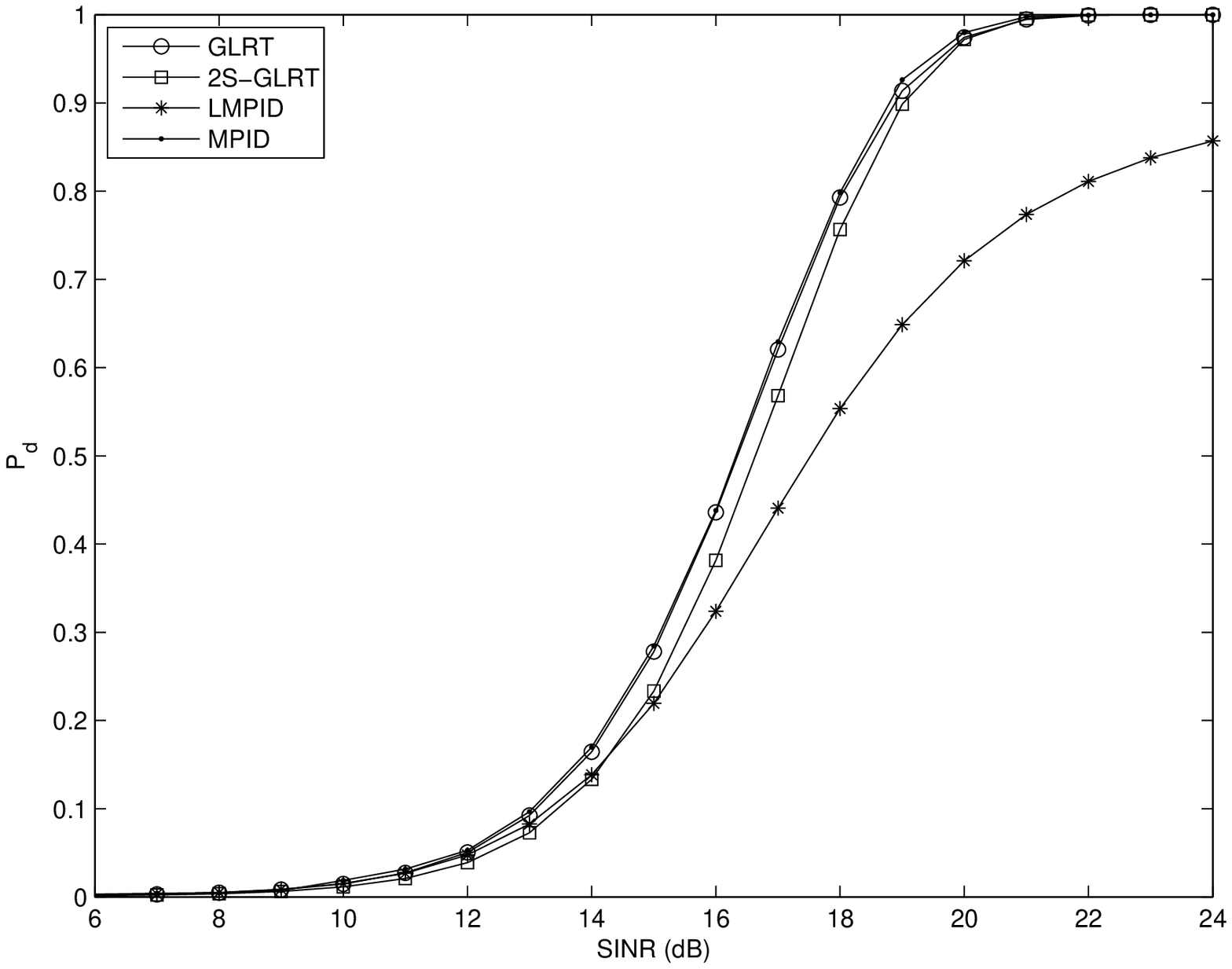}
    \caption{$P_d$ versus SINR for the GLRT, the 2S-GLRT, the LMPID, and the MPID assuming $N=8$, $K=12$, $r=4$, $t=2$, and $P_{fa}=10^{-4}$.}
    \label{fig:fig02}
\end{figure}

\begin{figure}[!htp]
    \centering
   \includegraphics[width=7cm]{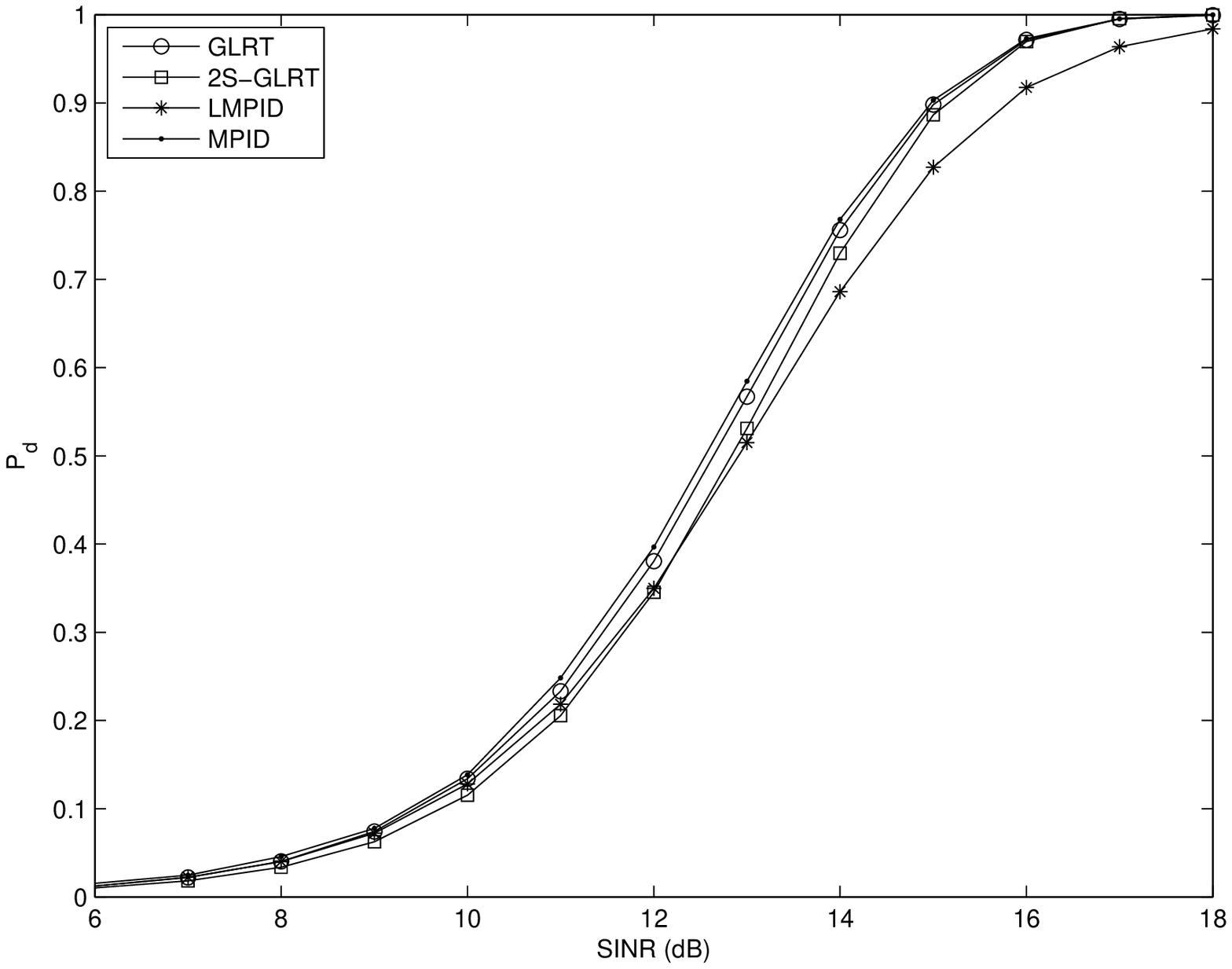}
    \caption{$P_d$ versus SINR for the GLRT, the 2S-GLRT, the LMPID, and the MPID assuming $N=8$, $K=16$, $r=2$, $t=4$, and $P_{fa}=10^{-4}$.}
    \label{fig:fig03}
\end{figure}

\begin{figure}[!htp]
    \centering
   \includegraphics[width=7cm]{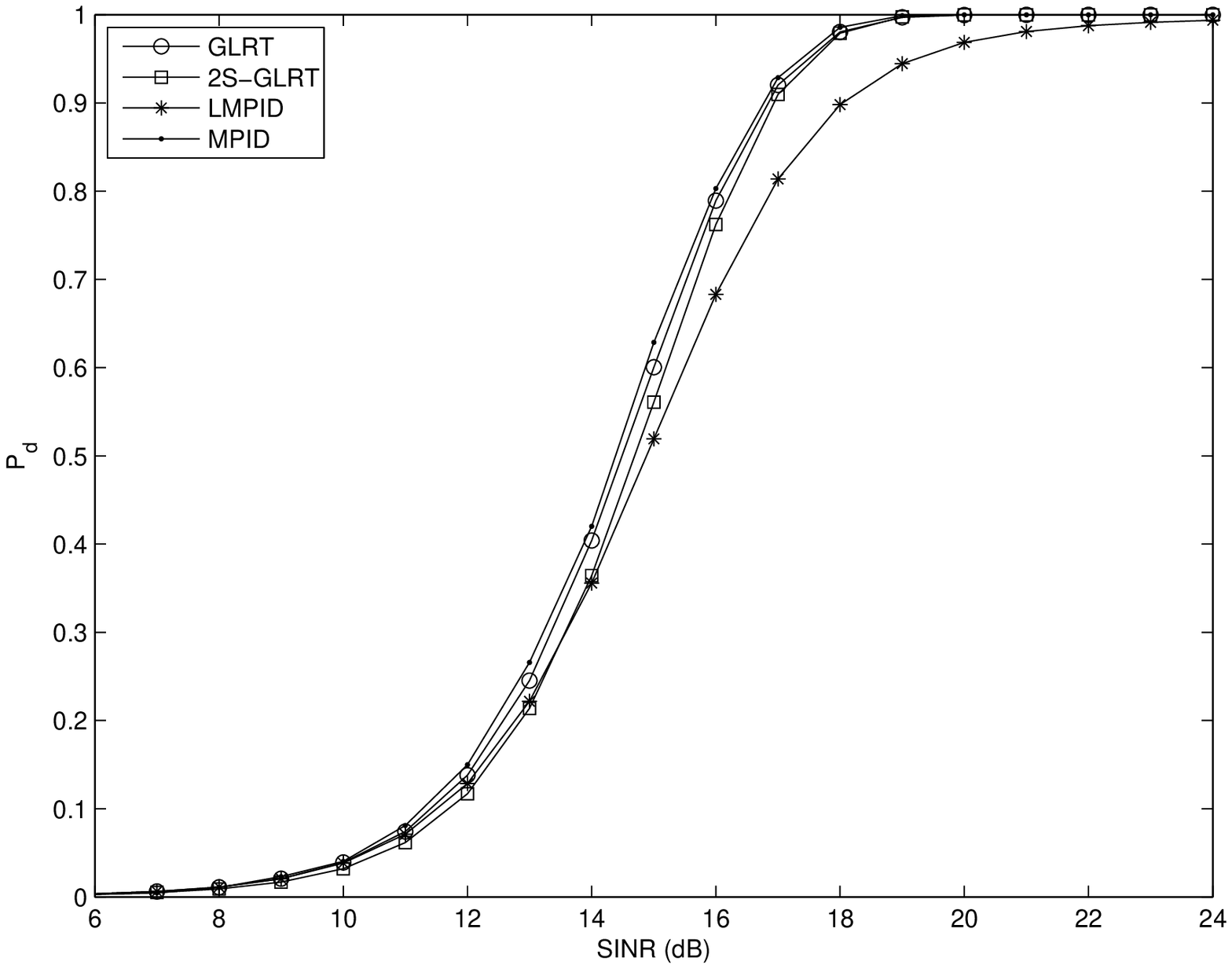}
    \caption{$P_d$ versus SINR for the GLRT, the 2S-GLRT, the LMPID, and the MPID assuming $N=8$, $K=16$, $r=4$, $t=2$, and $P_{fa}=10^{-4}$.}
    \label{fig:fig04}
\end{figure}

\end{document}

%% file: Defs2.tex
\newcommand{\bit}{\begin{itemize}}
\newcommand{\eit}{\end{itemize}}

\newcommand{\ben}{\begin{enumerate}}
\newcommand{\een}{\end{enumerate}}

\newcommand{\bdesc}{\begin{description}}
\newcommand{\edesc}{\end{description}}

\newcommand{\bea}{\begin{array}}
\newcommand{\eea}{\end{array}}

\newcommand{\beqa}{\begin{eqnarray}}
\newcommand{\eeqa}{\end{eqnarray}}

\newcommand{\ds}{\displaystyle}

\newcommand{\Comment}[1]{}

\newtheorem{prop}{Proposition}
\newtheorem{corol}{Corollary}

\def\N{{\mathds N}}

\def\R{{\mathds R}}

\def\C{{\mathds C}}

\def\cC{\mbox{$\mathcal C$}}
\def\cF{\mbox{$\mathcal F$}}
\def\cG{\mbox{$\mathcal G$}}
\def\cL{\mbox{$\mathcal L$}}
\def\cN{\mbox{$\mathcal N$}}

\def\cW{\mbox{$\mathcal W$}}

\newcommand{\be}{\begin{equation}}
\newcommand{\ee}{\end{equation}}

\newcommand{\bzero}{{\mbox{\boldmath $0$}}}

\newcommand{\bd}{{\mbox{\boldmath $d$}}}

\newcommand{\bff}{{\mbox{\boldmath $f$}}}
\newcommand{\bn}{{\mbox{\boldmath $n$}}}
\newcommand{\bm}{{\mbox{\boldmath $m$}}}
\newcommand{\bp}{\mbox{\boldmath $p$}}
\newcommand{\bor}{{\mbox{\boldmath $r$}}}
\newcommand{\bw}{{\mbox{\boldmath $w$}}}
\newcommand{\bs}{{\mbox{\boldmath $s$}}}
\newcommand{\bt}{{\mbox{\boldmath $t$}}}

\newcommand{\bx}{{\mbox{\boldmath $x$}}}
\newcommand{\by}{{\mbox{\boldmath $y$}}}
\newcommand{\bq}{{\mbox{\boldmath $q$}}}
\newcommand{\bz}{{\mbox{\boldmath $z$}}}

\newcommand{\bA}{{\mbox{\boldmath $A$}}}
\newcommand{\bB}{{\mbox{\boldmath $B$}}}

\newcommand{\bD}{{\mbox{\boldmath $D$}}}
\newcommand{\bE}{{\mbox{\boldmath $E$}}}

\newcommand{\bG}{{\mbox{\boldmath $G$}}}
\newcommand{\bH}{{\mbox{\boldmath $H$}}}
\newcommand{\bI}{{\mbox{\boldmath $I$}}}
\newcommand{\bJ}{{\mbox{\boldmath $J$}}}

\newcommand{\bM}{{\mbox{\boldmath $M$}}}

\newcommand{\bP}{{\mbox{\boldmath $P$}}}
\newcommand{\bQ}{{\mbox{\boldmath $Q$}}}
\newcommand{\bR}{{\mbox{\boldmath $R$}}}
\newcommand{\bS}{{\mbox{\boldmath $S$}}}

\newcommand{\bU}{{\mbox{\boldmath $U$}}}
\newcommand{\bV}{{\mbox{\boldmath $V$}}}
\newcommand{\bX}{{\mbox{\boldmath $X$}}}

\newcommand{\bW}{{\mbox{\boldmath $W$}}}
\newcommand{\bZ}{{\mbox{\boldmath $Z$}}}

\newcommand{\btheta}{{\mbox{\boldmath $\theta$}}}
\newcommand{\bTheta}{{\mbox{\boldmath $\Theta$}}}

\newcommand{\tGLRT}{t_{\mbox{\tiny GLRT}}}
\newcommand{\tGLRTb}{\bar{t}_{\mbox{\tiny GLRT}}}
\newcommand{\tTSGLRT}{t_{\mbox{\tiny 2S-GLRT}}}

\newcommand{\tMPID}{t_{\mbox{\tiny MPID}}}
\newcommand{\tCLRT}{t_{\mbox{\tiny CMPID}}}
\newcommand{\tLRT}{t_{\mbox{\tiny LRT}}}
\newcommand{\tLMPID}{t_{\mbox{\tiny LMPID}}}

\newcommand{\tED}{t_{\mbox{\tiny ED}}}

\newcommand{\test}{\mbox{$
\begin{array}{c}
\stackrel{ \stackrel{\textstyle H_1}{\textstyle >} }{
\stackrel{\textstyle <}{\textstyle H_0} }
\end{array}
$}}

\newcommand{\testINV}{\mbox{$
\begin{array}{c}
\stackrel{ \stackrel{\textstyle H_0}{\textstyle >} }{
\stackrel{\textstyle <}{\textstyle H_1} }
\end{array}
$}}

\newcommand{\testinv}{\mbox{$
\begin{array}{c}
\stackrel{ \stackrel{\textstyle H_0}{\textstyle >} }{
\stackrel{\textstyle <}{\textstyle H_1} }
\end{array}
$}}
